\documentclass[preprint]{elsarticleLESSDATE}

\usepackage{color}
\usepackage{graphicx}
\usepackage{enumerate}
\usepackage{amsthm}

\newtheorem{theorem}{Theorem}
\newtheorem{lemma}[theorem]{Lemma}
\newtheorem{corollary}[theorem]{Corollary}
\newtheorem{definition}[theorem]{Definition}

\newcounter{fig}

\begin{document}

\title{On the Mathematics of\\Data Centre Network Topologies}

\author{Iain A. Stewart\corref{aaa}}%
\ead{i.a.stewart@durham.ac.uk}

\cortext[aaa]{This work was supported by the UK Engineering and Physical Sciences Research Council (EPSRC) grant `Interconnection Networks: Practice unites with Theory (INPUT)' [grant number  EP/K015680/1].}
\cortext[bbb]{A preliminary version of this paper appeared as an extended abstract in the \emph{Proceedings of 20th International Symposium on Fundamentals of Computation Theory\/} (A. Kosowski, I. Walukiewicz, eds.), Gdansk, Poland, August 17-19 2015, Lecture Notes in Computer Science Volume 9210, Springer, 2015, 283-295.}
\address{School of Engineering and Computing Sciences, Durham University,\\Science Labs, South Road, Durham DH1 3LE, U.K.}

\begin{abstract}
The theory of combinatorial designs has recently been used in order to build switch-centric data centre networks incorporating a large number of servers, in comparison with the popular Fat-Tree data centre network. The construction employed, called the $3$-step method, revolves around an appropriately chosen (but relatively small) bipartite graph and a transversal design. In this paper,  we clarify and extend these recent results. In particular, we prove the following path diversity results: in a one-to-one context, we prove that in these data centre networks there are pairwise link-disjoint paths joining all the servers adjacent to some switch with all the servers adjacent to any other switch so that we retain control of the path lengths (these results are optimal in terms of the numbers of paths constructed and we prove that we have a wide choice of bipartite graph and transversal design to which we can apply the $3$-step method); and in a one-to-many context, we prove that there are pairwise link-disjoint paths from all the servers adjacent to some switch to any identically-sized collection of target servers where these target servers need not be adjacent to the same switch (again, we keep control of the path lengths). Our constructions and analysis are undertaken on bipartite graphs with the applications to data centre networks being easily derived. Our results strengthen the overall competitiveness of data centre networks constructed using the $3$-step method, in comparison with Fat-Tree data centre networks, and, more generally, show the potential of results and methodologies from combinatorics to data centre network design.
\end{abstract}


\maketitle

\section{Introduction}

\subsection{The data centre network context}

Data centres are expanding both in terms of their size and their importance as computational platforms for cloud computing, web search, social networking, and so on. There is an increasing demand that data centres incorporate more and more servers but so that overall computational efficiency is not compromised through excessive traffic. A key factor as to the eventual performance of a data centre is the \emph{data centre network\/} (\emph{DCN\/}); that is, the interconnection fabric of the servers and switches within the data centre. As we strive to incorporate more and more servers, new topologies are being developed so as to cope with the increase in scale and best utilize the additional computational power. It is with topological aspects of DCNs that we are concerned in this paper.

The traditional design of a DCN is \emph{switch-centric\/} so that the routing intelligence resides amongst the switches, with the servers behaving only as computational nodes. In switch-centric DCNs, there are no direct server-to-server links; only server-to-switch and switch-to-switch links. Switch-centric DCNs are traditionally tree-like with servers located at the `leaves' of the tree-like structure. Examples include ElasticTree \cite{HSM06}, VL2 \cite{GHJ09}, HyperX \cite{ABD09}, Portland \cite{MPF09}, and Flattened Butterfly \cite{AMW10}, although the dominating switch-centric DCN is Fat-Tree \cite{ALV08}. Whilst it is generally acknowledged that tree-like, switch-centric DCNs have their limitations when it comes to, for example, scalability, due to the size of routing tables at the switches, switch-centric DCNs remain popular and can usually be constructed from commodity hardware. A more recent paradigm, namely the \emph{server-centric\/} DCN, has emerged so that deficiencies of the tree-like, switch-centric DCNs might be ameliorated. Server-centric DCNs reflect that the routing intelligence resides within the servers with switches operating only as dumb crossbars. In server-centric DCNs there are only server-to-switch and server-to-server links. However, server-centric DCNs also suffer from deficiencies such as packet relay overheads caused by the need to route packets within the server; moreover, server-centric DCNs have yet to make it into the commercial mainstream (the reader is referred to \cite{LMV13} for an overview of the state of the art as regards DCN architectural design). It is with the construction of switch-centric DCNs that we are concerned here. 

It is difficult to design computationally efficient (switch-centric) DCNs so as to incorporate large numbers of servers as there are additional considerations to take into account. For example, switches and (especially) servers in data centres have a limited number of ports with a consequence being that the more servers there are, the greater the average or worst-case link-count between two distinct servers; hence, there is a packet latency overhead to be borne. Also, so as to better support routing, fault-tolerance, and load-balancing, we would prefer that there are numerous alternative paths within the DCN joining any two distinct servers; that is, that there is \emph{path diversity\/}. There are many other design parameters to bear in mind relating to, for example, incremental scalability, throughput, cost, oversubscription, energy consumption, latency, and security (see, for example, \cite{WXN12} for an overview). The upshot is that the DCN designer has to simultaneously secure a number of performance characteristics, some of which are competing against each other; this makes the DCN design space difficult to work in. 

\subsection{Using combinatorial designs to build DCNs}

A recent proposal in \cite{QFZ15} advocated the use of \emph{combinatorial design theory\/} in order to design switch-centric DCNs; these DCNs have beneficial properties as regards incorporating more servers and possessing path diversity yet it is possible to limit the worst-case link-length of server-to-server shortest paths (and so, ultimately, achieve better control over packet latency in a DCN). The use of combinatorial designs within the study of general interconnection networks is not new and originated in \cite{BBD95} where the targeted networks involved processors communicating via buses (the reader is referred to \cite{CDS99} for a range of applications of combinatorial design theory within computer science). A hypergraph framework was developed in \cite{BBD95} where the hypergraph nodes represent the processors and the hyperedges the buses. Likewise, an analogous framework was developed in \cite{QFZ15} but where the hypergraph nodes and edges both represent switches so that the pendant servers `hang off' some of the switches (we present a detailed description of this framework in Section~\ref{sec:composition}). In \cite{QFZ15}, the ubiquitous switch-centric Fat-Tree DCN from \cite{ALV08} was used as a yardstick against which to compare the new DCN designs developed in \cite{QFZ15} under the normalization that all DCNs are to have the same worst-case link-length of server-to-server shortest paths, namely $6$, as this equals the worst-case link-length of server-to-server shortest paths in the Fat-Tree DCN. It was shown that more servers can be incorporated within the new DCNs yet, crucially, the resulting DCNs have good path diversity. It is the algebraic properties (relating to symmetry and balance) possessed by transversal designs that enable the constructions and analysis as described in \cite{QFZ15}. One slight difficulty with the original and novel approach taken in \cite{QFZ15} is that some of the path diversity results derived there are incorrect (as we explain later in Section~\ref{sec:difficulties}).

\subsection{Our contribution}

In this paper we return to the framework of \cite{QFZ15} and formulate and prove path diversity results for the switch-centric DCNs constructed using the methods of that paper. As our concern is entirely with topological properties of DCNs, henceforth we abstract our DCNs as undirected graphs where the nodes are to represent servers and switches and the edges point-to-point links. The crux of the construction in \cite{QFZ15} is (essentially) to build a bipartite graph using a systematic method, called the $3$-step method, involving a different `base' bipartite graph and a transversal design, and to convert the resulting bipartite graph into switch-centric DCNs (in a variety of ways). After explaining how hypergraphs and transversal designs can all be considered as bipartite graphs in Section~\ref{sec:basics}, in Section~\ref{sec:3stepconstruction} we provide a detailed description of the 3-step framework from \cite{QFZ15} and explain how the bipartite graphs constructed are converted into switch-centric DCNs. Next, we revisit the results from \cite{QFZ15}. In particular, in Section~\ref{sec:onetoonepathdiv} we correct and extend the analysis in \cite{QFZ15} and affirm that using the $3$-step method from \cite{QFZ15}, we can build switch-centric DCNs: with many more servers than the Fat-Tree DCN yet so that, like the Fat-Tree, every server-to-server shortest path has length at most $6$; and so that (assuming some numeric conditions on the base bipartite graph and the transversal design) we can find pairwise link-disjoint paths from all of the servers adjacent to a particular switch to all of the servers adjacent to any other switch. Moreover, we provide an upper bound on the lengths of the paths constructed in terms of the diameter of the base bipartite graph (see Theorem~\ref{thm:main}). We also deal with a scenario missing from \cite{QFZ15} (see part (\emph{b\/}) of Theorem~\ref{thm:main}). As we explain, the general situation is more subtle than was assumed in \cite{QFZ15}.

The DCN path diversity, as we have described it above, comes about from building bipartite graphs (which are subsequently converted to DCNs) so that given any two distinct nodes, there are numerous node-disjoint paths joining these two nodes; that is, these bipartite graphs have one-to-one path diversity. In Section~\ref{sec:onetomanypathdiv}, we go on to show that we can actually build numerous edge-disjoint paths from a source node to different destination nodes in our bipartite graphs; that is, we have one-to-many path diversity. The DCNs obtained from these bipartite graphs are such that (assuming some numeric conditions on the base bipartite graph and the transversal design) we can find pairwise link-disjoint paths from all of the servers adjacent to some switch to any identically-sized collection of servers (irrespective of which switch they are adjacent to). Consequently, we show that our DCNs provide support for additional communication patterns that are prevalent within data centre networks. It should be noted that one-to-many and many-to-many communication patterns are commonplace in data centres; for example in applications involving MapReduce.

This paper is unashamedly theoretical. However, we demonstrate that not only is there interesting combinatorics within the practical world of DCN design but that combinatorial mathematics can potentially contribute to the DCN design space on a practical level. We feel that the mathematical aspects of DCNs have so far remained almost completely unexamined and we advocate a closer theoretical scrutiny of DCNs both as a model of computation and in relation to the vast swathes of research on general interconnection networks. We mention some practical considerations and directions for further research in the Conclusion.

\section{Basic Concepts}\label{sec:basics}

We begin by briefly reviewing some architectural aspects of switch-centric DCNs that are pertinent to our subsequent research. We then move on to the discrete structures featuring in \cite{BBD95,QFZ15}, namely hypergraphs, bipartite graphs, and transversal designs. So that we might fully describe and understand the constructions in \cite{BBD95,QFZ15}, as well as our own upcoming analysis of switch-centric DCNs, we eventually amalgamate hypergraphs, bipartite graphs, and transversal designs so that by the end of this section, we will have developed an encompassing bipartite graph framework for the design of switch-centric DCNs. General graph-theoretic concepts can be obtained in \cite{Die10}. 

\subsection{Switch-centric DCNs}

A switch-centric DCN is abstracted as a graph (which we also refer to as a DCN) where the nodes are partitioned into two sets: there are \emph{server-nodes\/}; and there are \emph{switch-nodes\/}. Of course, the server-nodes correspond to servers in the DCN and the switch-nodes to switches; note that immediately there are practical design limitations imposed by the number of \emph{ports\/} in a real switch and the number of \emph{NIC ports\/} in a real server (we sometimes talk of the number of ports of a switch-node rather than its degree). Furthermore, in switch-centric DCNs there are no links joining one server-node directly to another server-node (because all routing within a switch-centric DCN falls within the purview of the switches). Of concern to us in this paper will be incorporating a comparatively large number of server-nodes within our DCNs but so that the maximum length of a shortest path joining any two server-nodes, that is, the \emph{diameter\/} of the DCN, is kept within a given bound, where the \emph{length\/} of such a path is the number of distinct links on the path. Essentially, we will be comparing DCNs as to how many server-nodes they incorporate but when their diameters are normalized.

However, DCNs must also possess other properties to make them usable within a data centre context. For example, they also need to: be scalable and incrementally scalable (that is, have the capacity to cope with increases in components and data); have low message latency; provide for high overall throughput (under a range of traffic patterns); be able to tolerate (a limited number of) faults; be energy efficient; be both economically and physically viable; and support virtualization (that is, the partitioning of the DCN into virtual networks on a dynamic basis), amongst many other things. Support for some of these properties can be measured using graph theory; for example, the diameter of the DCN gives guidance as regards the expected message latency. Of particular interest to us will be \emph{path diversity\/} which we define (somewhat informally) as the capacity to send data without inducing additional congestion or so as to cope with existing congestion or faults. There are two contexts of interest to us: the \emph{one-to-one\/} (or \emph{unicast\/}) context, when a source server-node wishes to send data to a destination server-node by the utilization of independent paths (we will return to what we mean by `independent' soon); and the \emph{one-to-many\/} (or \emph{multicast\/}) context, when a source server-node wishes to send data to a number of destination server-nodes so that the different transmissions do not induce congestion. Path diversity is highly relevant to a number of the above properties such as latency and scalability, where different paths are used to split and balance loads, and fault tolerance, where different paths provide alternative means of transit in the case of faults. That this is the case in the one-to-one context is obvious; however, the need for data centres to support data replication and applications like MapReduce \cite{DG04} makes path diversity crucial in a one-to-many context too. As we shall soon see, just as with latency, the independence of paths can be considered graph-theoretically.

\subsection{Hypergraphs}\label{sec:hypergraphs}

Hypergraphs provide the original framework for the \emph{$3$-step construction\/} (to be defined later) as employed in \cite{BBD95,QFZ15}: in \cite{BBD95}, hypergraphs were used to model bus interconnection networks; and in \cite{QFZ15}, hypergraphs were used to model data centre networks. For the moment, and in order to appreciate the context of \cite{BBD95,QFZ15}, we retain this hypergraph framework before we phrase all content in this introduction within an encompassing bipartite graph framework.

A \emph{hypergraph\/} $H = (V,E)$ consists of a finite set $V$ of \emph{nodes\/} together with a finite set $E$ of \emph{hyperedges\/} where each hyperedge is a non-empty set of nodes and each node appears in at least one hyperedge. The \emph{degree\/} of a node is the number of hyperedges containing it and the \emph{rank\/} of a hyperedge is its size as a subset of $V$. A hypergraph is \emph{regular\/} (resp. \emph{uniform\/}) if every node has the same degree (resp. every hyperedge has the same rank) with this degree (resp. rank) being the \emph{degree\/} (resp. \emph{rank\/}) of the hypergraph. Every graph $G=(V,E)$ has a natural representation as a hypergraph: the nodes of the hypergraph are $V$; and the hyperedges are $E$, where the hyperedge $e$ consists of the pair of nodes incident with the edge $e$ of $G$.

\subsection{Hypergraphs and bipartite graphs}

We can represent any hypergraph $H=(V,E)$ as a bipartite graph: the node set of the bipartite graph is $V\cup E$; and there is an edge $(v,e)$, for $v\in V$ and $e\in E$, in the bipartite graph if, and only if, $v\in e$ in the hypergraph. It is clear that this yields a one-to-one correspondence between hypergraphs and bipartite graphs (without isolated nodes) that come complete with a partition of the elements into a `left-hand side', which will correspond to the nodes of the hypergraph, and a `right-hand side', which will correspond to the hyperedges of the hypergraph. We assume (henceforth) that every bipartite graph comes equipped with such a partition and for clarity from now on we refer to the nodes on the left-hand side as \emph{nodes\/} and the nodes on the right-hand side as \emph{blocks\/} (this is in keeping with our upcoming realisation of transversal designs as bipartite graphs). Likewise, we refer to the degree of a node as its \emph{degree\/} and the degree of a block as its \emph{rank\/}. A bipartite graph corresponding to a regular, uniform hypergraph of degree $d$ and rank $\Delta$ is called a \emph{$(d,\Delta)$-bipartite graph\/}. Every bipartite graph (and so every hypergraph) also describes its \emph{dual bipartite graph\/} (or alternatively its dual hypergraph) where the roles of the nodes on the left-hand side and the blocks on the right-hand side of the partition are reversed in the definition of the bipartite graph; so, for example, the dual bipartite graph of a $(d,\Delta)$-bipartite graph is regular of degree $\Delta$ and uniform of rank $d$. 

Note that if $G$ is a bipartite graph then it corresponds to a hypergraph via our representation above and it also corresponds to a hypergraph via the natural representation highlighted in Section~\ref{sec:hypergraphs}. The two hypergraphs corresponding to the same bipartite graph are different and we are never interested in the representation of a bipartite graph as a hypergraph via the natural representation of Section~\ref{sec:hypergraphs}.

\subsection{Paths in hypergraphs}

A \emph{path\/} in some hypergraph $H=(V,E)$ (or the corresponding bipartite graph) is an alternating sequence of nodes and hyperedges so that all nodes are distinct, all hyperedges are distinct, and a node $v\in V$ follows or preceeds a hyperedge $e\in E$ in the sequence only if $v\in e$ in the hypergraph (or $(v,e)$ is an edge in the corresponding bipartite graph). The first element of some path is the \emph{source\/} and the final element the \emph{destination\/}. The \emph{length\/} of any path is its length in the bipartite graph corresponding to the hypergraph, and the \emph{distance\/} between two distinct elements of $V\cup E$ is the length of a shortest path joining these two elements in the corresponding bipartite graph. The \emph{diameter\/} of $H$ is the maximum of the distances between every pair of distinct nodes of $V$, and the \emph{line-diameter\/} of $H$ is the maximum of the distances between every pair of distinct hyperedges of $E$.

We have two remarks. First, we have traditional notions of diameter and line-diameter in any bipartite graph. Note that our notion of diameter in a bipartite graph, which is the longest shortest node-to-node path (and so ignores node-to-block and block-to-block paths), is different from the usual graph-theoretic notion of diameter in a bipartite graph (the same comment can be made as regards line-diameter). When we talk of the diameter or line-diameter of a bipartite graph, we mean with respect to \emph{our\/} notion of diameter or line-diameter, respectively; if we need to talk of the traditional notion of graph diameter then we will make this clear. Second, our notion of path length in a hypergraph  differs from that in \cite{QFZ15} where the length is the number of nodes (resp. hyperedges) in a hyperedge-to-hyperedge (resp. node-to-node) path. There is no real consequence to this difference; essentially, our notion of path length is double that in \cite{QFZ15}. However, we shall soon move to an exclusively bipartite graph-theoretic formulation in which our notion of length is the natural one to adopt.

We shall be interested in building sets of paths in some hypergraph $H$ so that the paths might have the same sources or destinations; moreover, we shall require that these paths do not `interfere' with one another (or are `independent' as we mentioned earlier). We say that a set of paths in $H$ is:
\begin{itemize}
\item \emph{pairwise internally-disjoint\/} if any source or destination only appears as a source or destination, and any node or hyperedge that is not a source or destination appears on at most one path
\item \emph{pairwise edge-disjoint\/} if every pair $(v,e)\in V\times E$ is such that $v$ follows or precedes $e$ on some path at most once across all paths from this set.
\end{itemize}

\subsection{Hypergraphs as switch-centric DCNs}\label{sec:hyperDCNs}

Given some hypergraph $H=(V,E)$, our intention is to ultimately transform this hypergraph into a DCN by considering both the nodes and the hyperedges as switch-nodes so that the switch-nodes corresponding to the nodes (which we shall later call the level-$1$ switch-nodes, with the switch-nodes corresponding to the hyperedges the level $2$-switch-nodes) also have adjacent server-nodes, which we have yet to define (this intention is best appreciated by working with the corresponding bipartite graph rather than the hypergraph; the upcoming Fig.~\ref{MethodAce1} provides a visualization of what we mean). Consequently, we can regard a hypergraph $H$ as modelling a switch-centric DCN $N$ where there are two levels of switch-nodes. 

Suppose that we have a set of pairwise internally-disjoint paths from one node of $H$ to another node of $H$. This translates to a set of pairwise internally-disjoint paths in $N$ from a level-$1$ switch-node to another level-$1$ switch-node. We can use these paths for the simultaneous transfer of data from  server-nodes adjacent to the source level-$1$ switch-node to server-nodes adjacent to the destination level-$1$ switch-node. In order to facilitate this data transfer we require that level-$1$ switch-nodes are non-blocking whereas the level-$2$ switch-nodes can be blocking; recall that a switch-node is \emph{non-blocking\/} when no contention arises when simultaneously sending data through the switch-node on two distinct input links and out on two distinct output links, and \emph{blocking\/} otherwise. If our paths in $H$ are only pairwise edge-disjoint then we require that all switch-nodes of $N$ are non-blocking.

\subsection{Transversal designs}\label{sec:tdesigns}

The notion of a transversal design is crucial to what follows.

\begin{definition}\label{def:transversal} Let $k,\Delta\geq 2$. A \emph{$[\Delta,k]$-transversal design\/} $T$ is a triple $(\mathcal{X},\mathcal{D},\mathcal{U})$ where: $|\mathcal{X}|=\Delta k$; $\mathcal{D} = (D_1,D_2,\ldots,D_\Delta)$ is a partition of $\mathcal{X}$ into $\Delta$ equal-sized \emph{groups\/} (each of size $k$); and $\mathcal{U}=\{U_j:j=1,2,\ldots,k^2\}$ is a family of $k^2$ subsets of $\mathcal{X}$, each of size $\Delta$ and called a \emph{block\/}, so that
\begin{itemize}
\item $|D_i\cap U_j| = 1$, for $i=1,2,\ldots, \Delta$, $j=1,2,\ldots,k^2$
\item each pair of elements $\{x_i,x_j\}$, where $x_i\in D_i$, $x_j\in D_j$ and $i\neq j$, is contained in exactly $1$ block.
\end{itemize}
\end{definition}

We adopt a graph-theoretic perspective on transversal designs as defined in Definition~\ref{def:transversal}: we think of the $[\Delta,k]$-transversal design $T$ as a bipartite graph where the elements of $\mathcal{X}$ (resp. $\mathcal{U}$) lie on the left-hand side (resp. right-hand side) of the partition, and so are called nodes (resp. blocks) within the bipartite graph, and so that in this bipartite graph there is an edge $(p,Q)$, for $p\in \mathcal{X}$ and $Q\in \mathcal{U}$, if, and only if, in the transversal design the element $p$ is in the block $Q$. Note that the bipartite graph corresponding to the transversal design from Definition~\ref{def:transversal} is a $(k,\Delta)$-bipartite graph. Henceforth, we adopt our bipartite graph framework and regard both hypergraphs and transversal designs as bipartite graphs (unless we state otherwise).

There is an intimate relationship involving transversal designs, \emph{orthogonal arrays\/} and \emph{mutually orthogonal latin squares\/}, although there is no need to give definitions here. However, it is well known: that there are $\Delta$ mutually orthogonal latin squares of order $k$ if, and only if, there is a $[\Delta+2,k]$-orthogonal array if, and only if, there is a $[\Delta+2,k]$-transversal design; and that there are at most $k-1$ mutually orthogonal latin squares of order $k$ (see, for example, \cite{Sti04}). Hence, if we have a $[\Delta,k]$-transversal design then $\Delta\leq k+1$. Also, if $k$ is a prime power then a $[\Delta,k]$-transversal design exists whenever $2\leq \Delta \leq k+1$ (again, see \cite{Sti04}). We shall use these facts later on. The study of the existence of $[\Delta,k]$-transversal designs, for various $\Delta$ and $k$, is a long-standing area of research.

We require one final bit of notation. If $T$ is some transversal design, as in Definition~\ref{def:transversal}, and $x$ and $y$ are nodes in distinct groups then we refer to the unique block adjacent to both $x$ and $y$ as the block \emph{generated\/} by $x$ and $y$.

\section{The $3$-step Construction and its Extensions}\label{sec:3stepconstruction}

We now describe the \emph{$3$-step construction\/} for building bipartite graphs (or, equivalently, hypergraphs) by using a `base' bipartite graph and a transversal design (which we think of as a bipartite graph). This construction originated in \cite{BBD95} and was used in \cite{QFZ15}. We then explain how this construction was subsequently extended in \cite{QFZ15} both by iteration and by composition so as to yield switch-centric DCNs.

\subsection{The $3$-step construction}\label{subsec:3step}

The $3$-step construction proceeds as follows.\smallskip

\noindent\underline{Step 1}: Let $H_0$ be a connected $(d,\Delta)$-bipartite graph so that there are $n$ nodes (on the left-hand side of the partition, each of degree $d$) and $e$ blocks (on the right-hand side, each of rank $\Delta$). Such an $H_0$ can be visualized as in Fig.~\refstepcounter{fig}\thefig\label{dDeltaBipartite} (ordinarily, we represent nodes as circles and blocks as squares).\smallskip

\begin{figure}[t]
\centering
\scalebox{0.68}[0.68]{
\includegraphics{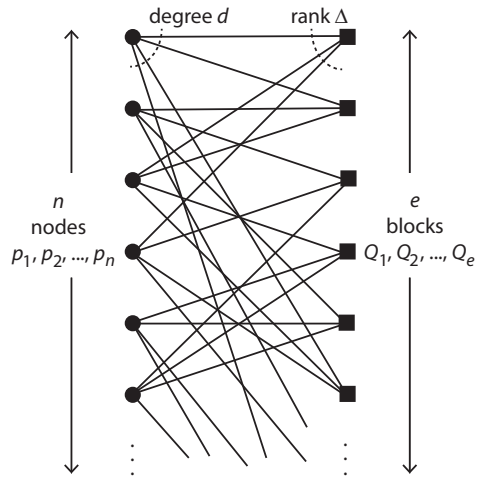}}
\caption{A $(d,\Delta)$-bipartite graph $H_0$.}
\end{figure}

\noindent\underline{Step 2}: Let $T$ be a $[\Delta,k]$-transversal design. In particular, there are $\Delta$ groups of $k$ nodes (on the left-hand side) as well as $k^2$ blocks (on the right-hand side). Such a $T$ can be visualized as in Fig.~\refstepcounter{fig}\thefig\label{DeltakTransversal}. Build the bipartite graph $H$ as follows. For every node $p$ of $H_0$, introduce a group $G_p$ of $k$ nodes of $H$; we say that the group of nodes $G_p$ of $H$ is \emph{associated with\/} the node $p$ of $H_0$. For every block $Q$ of $H_0$, adjacent to the nodes $p_1,p_2,\ldots,p_\Delta$ in $H_0$, introduce a copy of $T$, denoted $T_Q$, rooted on the $\Delta$ groups of nodes $G_{p_1},G_{p_2},\ldots,G_{p_\Delta}$; so, \emph{associated with\/} the block $Q$ of $H_0$, we have a set $B_Q$ of $k^2$ blocks in $H$. We refer to the $\Delta$ groups of nodes $G_{p_1},G_{p_2},\ldots,G_{p_\Delta}$ as the \emph{roots\/} of the copy $T_Q$ of $T$ in $H$. Such a bipartite graph $H$ can be visualized as in Fig.~\refstepcounter{fig}\thefig\label{3StepTwo} where two of the copies of $T$ are partially shown (note that they might have some roots in common but their respective sets of blocks are always disjoint as are their sets of edges). The bipartite graph $H_0$ provides a template as to how we introduce copies of $T$ to form $H$.

\begin{figure}[t]
\centering
\scalebox{0.68}[0.68]{
\includegraphics{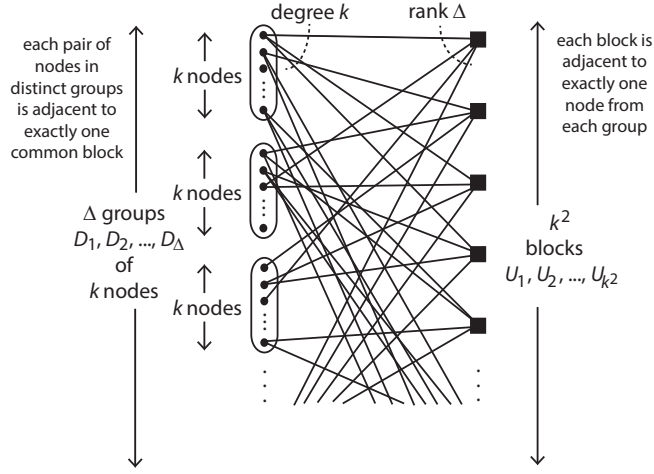}}
\caption{A $[\Delta,k]$-transversal design $T$.}
\end{figure}

\begin{figure}[t]
\centering
\scalebox{0.68}[0.68]{
\includegraphics{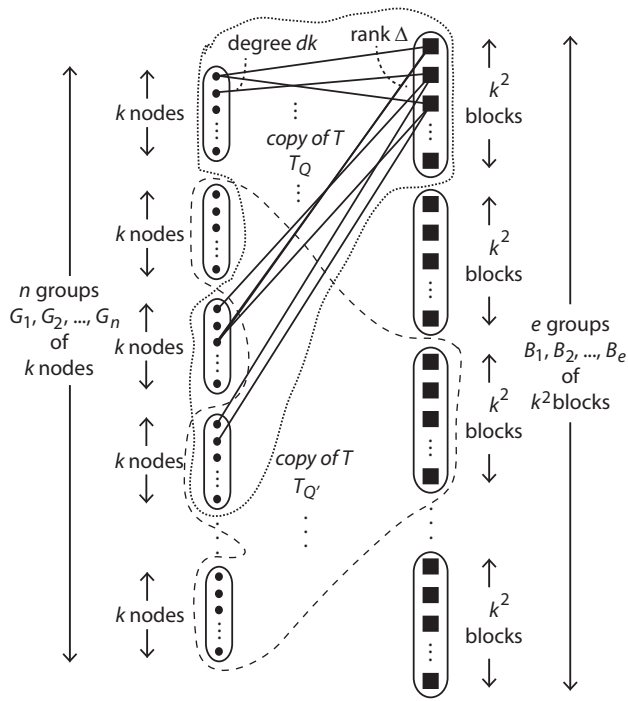}}
\caption{Amalgamating $H_0$ and $T$ to get $H$.}
\end{figure}

Note that:
\begin{itemize}
\item each node of $H$ can be indexed as $a_{p,j}$, where $p\in\{p_i:i=1,2,\ldots,n\}$ and $j\in\{1,2,\ldots,k\}$, so that $p$ is the node of $H_0$ to which the group $G_p$ in which $a_{p,j}$ sits is associated and $j$ is the index of the node $a_{p,j}$ in this group
\item each block of $H$ can be indexed as $B_{Q,U}$, where $Q\in\{Q_i:i=1,2,\ldots,e\}$ and $U\in\{1,2,\ldots,k^2\}$, so that $Q$ is the block of $H_0$ to which the set of blocks $B_Q$ in which $B_{Q,U}$ sits is associated and $U$ is the block of $T$ to which $B_{Q,U}$ corresponds.
\end{itemize}
In addition, each node of $T$ can be indexed $u_{i,j}$, where $i\in\{1,2,\ldots,\Delta\}$ and $j\in\{1,2,\ldots,k\}$, so that $D_i$ is the group of nodes in which $u_{i,j}$ sits and $j$ is the index of $u_{i,j}$ in that group.\smallskip

\noindent\underline{Step 3}: Let $H^\ast$ be the bipartite graph obtained from the bipartite graph $H$ by reversing the roles of nodes and blocks (so, $H^\ast$ is the dual bipartite graph of $H$). Note that the bipartite graph $H^\ast$ is regular of degree $\Delta$ and uniform of rank $dk$.\smallskip

We refer to the $(dk,\Delta)$-bipartite graph $H$ (resp. the $(\Delta,dk)$-bipartite graph $H^\ast$) constructed above as having been constructed by the $2$-step (resp. $3$-step) method using the $(d,\Delta)$-bipartite graph $H_0$ and the $[\Delta,k]$-transversal design $T$. Note that $H$ (resp. $H^\ast$) has $nk$ nodes (resp. $ek^2$ nodes) and $ek^2$ blocks (resp. $nk$ blocks).

Our intention with our constructions is to ultimately design switch-centric DCNs with beneficial properties (as we outlined in Section~\ref{sec:hyperDCNs}). Whilst there are many properties we would like our DCNs to have, it is important that DCNs can integrate a large number of server-nodes so that the server-node-to-server-node distances are short and so that there is redundancy as to which (short) server-node-to-server-node routes we choose to use. In our framework of bipartite graphs, this translates as building bipartite graphs with a large number of nodes and with redundant (short) node-to-node paths. As a first step, the following result was proven in \cite{BBD95} (it is actually derivable from the proofs of our upcoming results) and allows us control over the length of shortest block-to-block paths in $2$-step constructions (and so shortest node-to-node paths in $3$-step constructions).

\begin{theorem}[\cite{BBD95}]\label{thm:linediameter}
Suppose that the $(dk,\Delta)$-bipartite graph $H$ has been constructed using the $2$-step method using the $(d,\Delta)$-bipartite graph $H_0$ and the $[\Delta,k]$-transversal design $T$. If $H_0$ has line-diameter $\lambda\geq 4$ then $H$ has line-diameter $\lambda$.
\end{theorem}
Of course, if $H^\ast$ is the dual bipartite graph of $H$ in Theorem~\ref{thm:linediameter} then it has diameter $\lambda$. We reiterate that our notion of diameter and line-diameter differs from that in \cite{BBD95,QFZ15} (where the length of a block-to-block path is the number of nodes on that path; so, in \cite{BBD95,QFZ15} the bound $\lambda\geq 4$ in our Theorem~\ref{thm:linediameter} appears as $\lambda\geq 2$).

\subsection{Iteration}

We can iterate the $3$-step construction (as was done in \cite{QFZ15}). Note that if $H_0$ is a $(d,\Delta)$-bipartite graph of line-diameter $\lambda\geq 4$, with $n$ nodes and $e$ blocks, then the bipartite graph $H_1$ resulting from the $2$-step construction (using $H_0$ and some $[\Delta,k]$-transversal design $T$) is a $(dk,\Delta)$-bipartite graph of line-diameter $\lambda$. So, repeating the $2$-step construction but with $H_1$ replacing $H_0$ (we keep the same $T$, although we do not have to) yields a $(dk^2,\Delta)$-bipartite graph $H_2$ of line-diameter $\lambda$. By iterating this construction, we can clearly obtain a $(dk^i,\Delta)$-bipartite graph $H_i$ of line-diameter $\lambda$. Converting $H_i$ into $H_i^\ast$ results in a bipartite graph with $ek^{2i}$ nodes, with $nk^i$ blocks, with diameter $\lambda$, and that is regular of degree $\Delta$ and uniform of rank $dk^i$. 

\subsection{Composition}\label{sec:composition}

New methods of composing bipartite graphs (built according to the $3$-step construction) so as to obtain switch-centric DCNs were also derived in \cite{QFZ15}. In \cite{QFZ15}, $4$ such methods were given: Methods $M_1$, $M_2$ and $M_3$ are different cases of Method $A$, below; and Method $M_4$ is Method $B$.

In what follows, let $H$ be a $(\Delta,\delta)$-bipartite graph where $\Delta < \delta$ and where there are $n$ nodes and $e$ blocks.\smallskip

\noindent\underline{Method $A$}: We take $c$ copies of $H$ where $\delta - c\Delta > 0$ and $c\geq 1$. For each node $u$ of $H$: we remove the corresponding node in each of the $c$ copies of $H$ and introduce a new switch-node (common to all copies of $H$); we make all of the $c\Delta$ edges incident with the $c$ original nodes incident with this new switch-node; and we attach $\rho=\delta-c\Delta$ pendant server-nodes to the new switch-node. All blocks of $H$ are considered as switch-nodes. We follow \cite{QFZ15} and call the new switch-nodes \emph{level-$1$ switch-nodes\/} and the original switch-nodes \emph{level-$2$ switch-nodes\/}. The construction of the switch-centric DCN $N(H)$ from $H$ via this method can be visualised as in Fig.~\refstepcounter{fig}\thefig\label{MethodAcgte2}, where we only show the construction for the $c$ nodes coresponding to one node of $H$. Note that every switch-node of $N(H)$ has $\delta$ ports. Also, there is some choice as regards the parameter $c$ (so that choosing different values for $c$ yields different values for $\rho$). We illustrate the special case when $c=1$ in Fig.~\refstepcounter{fig}\thefig\label{MethodAce1}, where $H$ is a $(3,5)$-bipartite graph. The general case when $c\geq 1$ corresponds to Method $M_2$ of \cite{QFZ15}; the special case when $c = 1$ corresponds to Method $M_1$; and the special case when $c=\lfloor\frac{\lfloor\frac{\delta}{2}\rfloor}{\Delta}\rfloor$ corresponds to Method $M_3$. In this latter case, the aim is to ensure that every level-$1$ switch-node is adjacent to roughly the same number of level-$2$ switch-nodes as it is server-nodes. Note that: the number of server-nodes in $N(H)$ is $n(\delta-c\Delta)$; the number of level-$1$ switch-nodes is $n$; and the number of level-$2$ switch-nodes is $ce$.\smallskip

\begin{figure}[t]
\centering
\scalebox{0.68}[0.68]{
\includegraphics{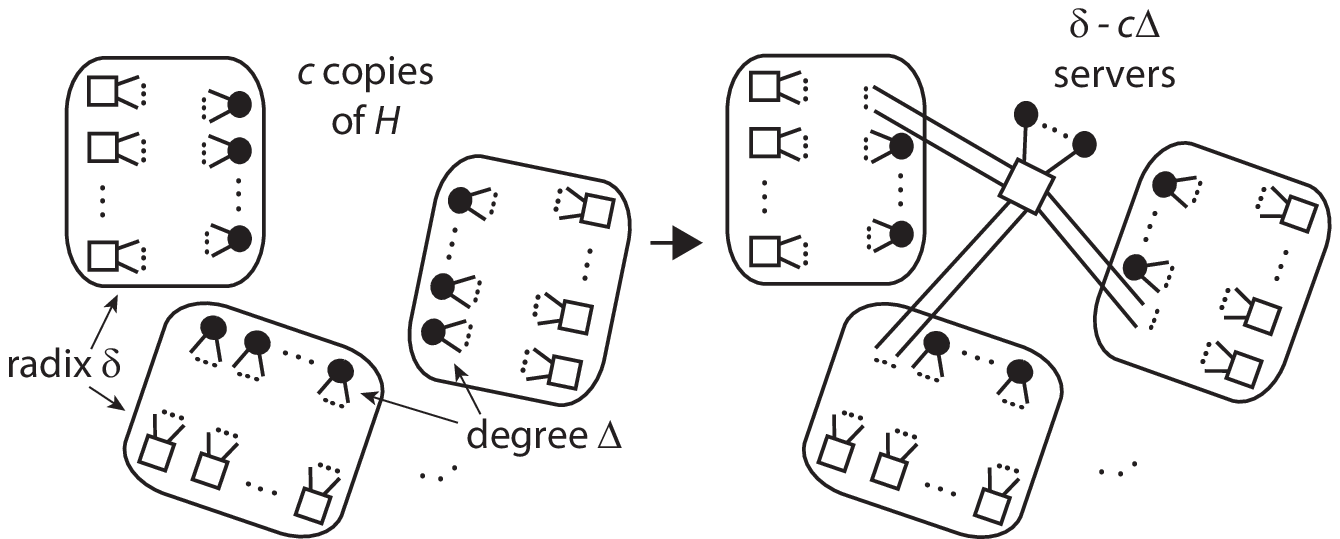}}
\caption{Building a switch-centric DCN via Method $A$ when $c>1$.}
\end{figure}

\begin{figure}[t]
\centering
\scalebox{0.68}[0.68]{
\includegraphics{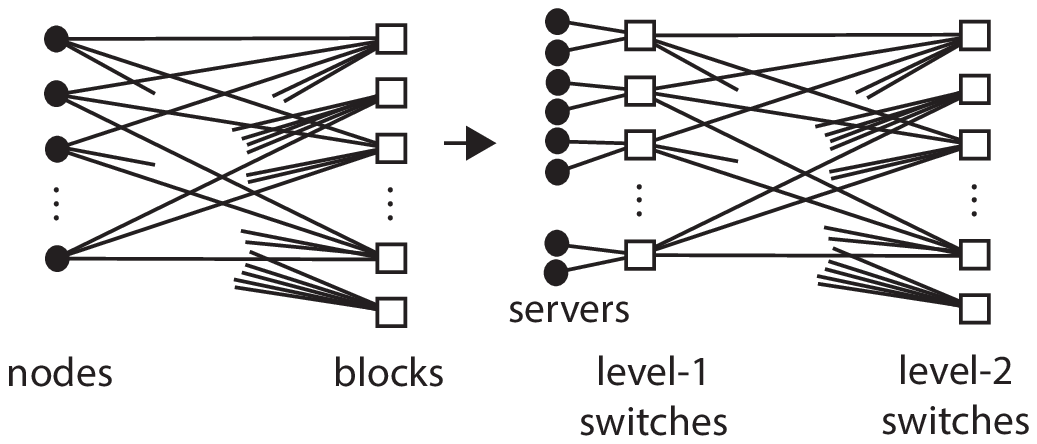}}
\caption{Building a switch-centric DCN via Method $A$ when $c=1$.}
\end{figure}

\noindent\underline{Method $B$}: We now work with a switch-centric DCN as constructed by Method $A$. Let every level-$1$ switch-node have $\rho$ adjacent server-nodes. Suppose that there is an even number of level-$1$ switch-nodes. Partition the set of level-$1$ switch-nodes into pairs. For each pair of switch-nodes $(S^\prime,S^{\prime\prime})$: remove $\lfloor\frac{\rho}{2}\rfloor$ server-nodes that are adjacent to $S^\prime$ and remove $\lceil\frac{\rho}{2}\rceil$ server-nodes that are adjacent to $S^{\prime\prime}$; and make every server-node that is adjacent to the switch-node $S^\prime$ or the switch-node $S^{\prime\prime}$ also adjacent to the other switch-node. Note that the number of ports of any switch-node has not changed but that every server-node is now adjacent to $2$ switch-nodes. The philosophy behind this construction is to better tolerate the failure of a level-1 switch-node. The construction can be visualized as in Fig.~\refstepcounter{fig}\thefig\label{MethodB} where paired level-$1$ switch-nodes have the same shade of grey and where $\rho=3$.\smallskip

\begin{figure}[t]
\centering
\scalebox{0.68}[0.68]{
\includegraphics{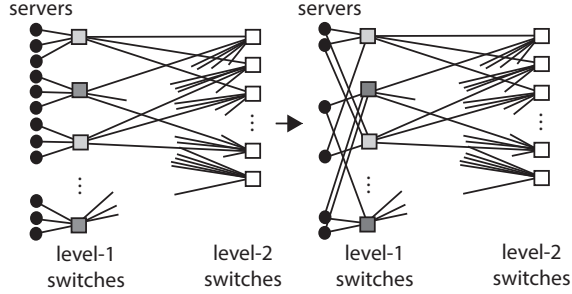}}
\caption{Building a switch-centric DCN via Method $B$.}
\end{figure}

\subsection{Some illustrations of DCNs}

In \cite{QFZ15}, switch-centric DCNs constructed using the $3$-step method allied with Methods $A$ and $B$ were favourably compared with the $3$-level Fat-Tree DCN from \cite{ALV08} with regard to the number of server-nodes therein when the diameter and the number of ports of a switch-node are held constant. The reader is referred to \cite{ALV08,QFZ15} for full details as regards the topology of Fat-Tree and to Tables 2--4 in \cite{QFZ15} for the complete comparison; however, we include a replicated table here purely for illustrative purposes. In Table~\ref{tab:comparisontable} (which is Table~2 from \cite{QFZ15}): the number of ports of any switch-node is forced to be $64$; the diameters of the DCNs resulting from using the $3$-step method, iteration and composition are forced to be (at most) $6$ (like that of Fat-Tree); and the numbers of server-nodes and switch-nodes in the resulting DCNs are as given (note that the length of a server-node-to-server-node path as defined in \cite{QFZ15} is the number of switch-nodes on it, which is one less than our notion of length which is the number of links on the path). 

\begin{table}[t]
\setlength{\tabcolsep}{5.5pt}
\caption{Comparing switch-centric DCNs built with switch-nodes with $64$ ports}
\centering
\begin{tabular}{|l| c| c| c| c|}
\hline
 network & \# switch ports & diameter & \# server-nodes & \# switch-nodes \\
\hline
 Fat-Tree & $64$ & $6$ & $65,536$ & $5,120$\\
 $H^\ast$ & $64$ & $4$ & $54,720$ & $6,840$\\
 $N_A^1(H^\ast)$ & $64$ & $6$ & $3,064,320$ & $61,560$\\
 $N_A^2(H^\ast)$ & $64$ & $6$ & $437,760$ & $102,600$\\
 $N_A^3(H^\ast)$ & $64$ & $6$ & $1,751,040$ & $82,080$\\
 $N_B(H^\ast)$ & $64$ & $6$ & $1,532,160$ & $61,560$\\
 $\bar{H}^\ast$ & $64$ & $4$ & $20,480$ & $1,280$\\
 $N_A^1(\bar{H}^\ast)$ & $64$ & $6$ & $1,228,800$ & $21,760$ \\
\hline
\end{tabular}
\label{tab:comparisontable}
\end{table}

\begin{itemize}

\item The bipartite graph $H^\ast$ is obtained using the $3$-step method starting with a $(8,8)$-bipartite graph $H_0$, that has $855$ nodes, $855$ blocks, and diameter and line-diameter $4$ (such a bipartite graph $H_0$ exists; see \cite{MS13}), and a $[8,8]$-transversal design $T$. The DCN $H^\ast$ in Table~\ref{tab:comparisontable} is the DCN obtained by simply regarding every node of the bipartite graph $H^\ast$ as a server-node (note that in this DCN we require that every server-node has $8$ NIC ports); the DCN $N_A^1(H^\ast)$ (resp. $N_A^2(H^\ast)$, $N_A^3(H^\ast)$) is obtained by employing Method $A$ with $c=1$ (resp. $c=7$, $c=4$); and the DCN $N_B(H^\ast)$ is obtained by employing Method $B$ with $N_A^1(H^\ast)$ (note that the number of switch-nodes entry in Table~2 in \cite{QFZ15} is incorrect). 
 
\item The bipartite graph $\bar{H}$ is obtained using the $3$-step method iterated twice, starting with a $(4,4)$-bipartite graph $\bar{H}_0$, that has $80$ nodes, $80$ blocks, and diameter and line-diameter $4$ (such a bipartite graph $\bar{H}_0$ exists; see \cite{MS13}), and a $[4,4]$-transversal design  $\bar{T}$ (actually, in \cite{QFZ15} this transversal design is not mentioned; it does, however, exist). The DCN $\bar{H}^\ast$ in Table~\ref{tab:comparisontable} is the DCN obtained by simply regarding every node of the bipartite graph $\bar{H}^\ast$ as a server-node (note that the number of server-nodes entry in Table~2 in \cite{QFZ15} is incorrect, though the correct number is stated in the text); and the DCN $N_A^1(H^\ast)$ is obtained by employing Method $A$ with $c=1$ (note that the numbers of server-nodes and of switch-nodes entries in Table~2 in \cite{QFZ15} are incorrect).

\end{itemize}
It is clear from Table~\ref{tab:comparisontable} (and from \cite{QFZ15}) that we can build much bigger server-centric DCNs using the $3$-step method and the subsequent iterations and compositions than Fat-Tree but without increasing the diameter (which is a proxy for latency); of course, we would wish the new DCNs to have other properties that make them attractive within a data centre context. Establishing such properties was essentially the whole point of \cite{QFZ15} and we continue with this line of research in what follows.

Before we move to our main results, let us comment on using the $2$-step method as opposed to the $3$-step method when building our switch-centric DCNs (the same comment was made in \cite{QFZ15}). Note that when one uses the (iterated) $2$-step method, whilst the rank of the resulting bipartite graph stays the same, the degree grows. Were we to attach server-nodes to the switch-nodes that replace the nodes of the $2$-step bipartite graph $H$, rather than the $3$-step bipartite graph $H^\ast$, the number of ports of the level-$2$ switch-nodes (which would be $\Delta$) would be much less than the number of ports of the level-$1$ switch-nodes. Hence, it makes more sense to proceed as we have done above.

\section{One-to-one path diversity}\label{sec:onetoonepathdiv}

So far, we have set the scene from \cite{QFZ15} and described a method by which we can build bipartite graphs (the $3$-step method) which can then be transformed into switch-centric DCNs with many more servers than Fat-Tree whilst maintaining the diameter of Fat-Tree, i.e., $6$. However, as we mentioned earlier, there are many more aspects to the design of DCNs with an important one being path diversity. In what follows, we highlight some problems with the proofs of one-to-one path diversity in \cite{QFZ15} for bipartite graphs built using the $3$-step method. We then provide not only correct proofs as regards one-to-one path diversity but we also extend and improve the analysis in \cite{QFZ15} with new results. We end the section by  applying our constructions so as to build DCNs with good one-to-one path diversity properties.

\subsection{Difficulties with proofs}\label{sec:difficulties}

In order to detail the difficulties in \cite{QFZ15}, we adopt the terminology of \cite{QFZ15}. There are slight problems with the proof of Theorem 2 in \cite{QFZ15} (although they are easily surmountable). For example, in Subcases (1.2) and (2.2), $\{r_i,s_i,t_i\}\subseteq G_i^E$ and consequently we cannot generate the blocks $R_j$ and $S_j$. Also, in Subcase (2.1), the situation where $q\in P\cap Q\setminus\{p\}$ is not considered; it could be that $r_j=s_j$, for some $j\neq i$.

An attempt was also made in \cite{QFZ15} to extend Theorem 2 of \cite{QFZ15}: see Theorem~3 of \cite{QFZ15}. Assumptions concerning the connectivity of $H_0$ are made and the existence of additional paths in $H^\ast$ to those constructed in the proof of Theorem~2 are claimed in the situation when the two blocks $B_{Q,U}$ and $B_{Q^\prime,U^\prime}$ are such that $Q\neq Q^\prime$ (recall our method of indexing in Section~\ref{subsec:3step} which we adopt here). However, there are serious flaws in the proof of Theorem~3 of \cite{QFZ15}, so much so that the theorem is untrue. In short, Theorem~3 of \cite{QFZ15} claims that if there are $\omega$ pairwise internally-disjoint paths in $H_0$ from $Q$ to $Q^\prime$ then there are $\min\{\Delta\omega,k\omega\}$ pairwise internally-disjoint paths in $H$ from $B_{Q,U}$ to $B_{Q^\prime,U^\prime}$. This does not make sense: the maximum number of pairwise internally-disjoint paths in $H$ from $B_{Q,U}$ to $B_{Q^\prime,U^\prime}$ is $\Delta$ (as the bipartite graph $H$ has rank $\Delta$) and so we must have that $\min\{\Delta\omega,k\omega\} \leq \Delta$. For instance, in Example 1 of \cite{QFZ15}, the bipartite graph $H_0$ is the cycle of length $10$ ($H_0$ is derived from the cycle of length $5$ using its natural representation as a hypergraph; see Section~\ref{sec:hypergraphs}), so that $d=\Delta=2$, $n=e=5$, and there $2$ internally-disjoint paths from any block of $H_0$ to any other block of $H_0$. A $[2,3]$-transversal design $T$ is used and the bipartite graph $H^\ast$ built by the $3$-step method has rank $6$ and degree $2$. However, if Theorem~3 of \cite{QFZ15} were true then there would be $4$ pairwise disjoint paths from $B_{Q,U}$ to $B_{Q^\prime,U^\prime}$ in $H^\ast$ which clearly cannot be the case.

\subsection{The one-to-one scenario}

We now resurrect (some of) the proofs of the main results from \cite{QFZ15} and extend the results claimed in that paper. The following lemma proves most useful.

\begin{lemma}\label{lem:genblocks}
Let $T$ be some $[\Delta,k]$-transversal design with groups of nodes $\{D_1,$\linebreak$D_2,\ldots,D_\Delta\}$. Let $U$ be some block of $T$. For each $i\in\{1,2,\ldots,\Delta\}$, let $r_i\in D_i$ be the unique node of $D_i$ that is adjacent to $U$. Set $R=\{r_i:i=1,2,\ldots,\Delta\}$. Let $P$ be a set of distinct pairs of nodes so that: exactly one node of any pair in $P$ is in $R$ and no node of $R$ is in more than one pair of $P$; and no pair in $P$ is such that both nodes lie in the same group. The blocks generated by the pairs in $P$ are all distinct and different from $U$.
\end{lemma}

\begin{proof}
Suppose that $\{r_i,x\}\in P$, where $x\in D_l\setminus R$ with $l\neq i$ and where $i\in\{1,2,\ldots,\Delta\}$. Let $U_{r_i,x}$ be the block generated by $r_i$ and $x$. If $U_{r_i,x}=U$ then $U$ is adjacent to the distinct nodes $r_l$ and $x$ in $D_l$ which yields a contradiction.

Suppose that $\{r_j,y\}\in P\setminus\{\{r_i,x\}\}$, where $j\in\{1,2,\ldots,\Delta\}$. Let $U_{r_j,y}$ be the block generated by $r_j$ and $y$. Suppose that $U_{r_i,x}=U_{r_j,y}$; hence, $U_{r_i,x}$ is adjacent to both $r_i$ and $r_j$ with $i\neq j$. As any two nodes lying in distinct groups in $T$ are adjacent to a unique block of $T$, we must have that $U_{r_i,x}=U_{r_j,y}=U$; but this yields a contradiction as above. Hence, the blocks generated by the pairs in $P$ are all distinct and all different from $U$.
\end{proof}

We use this lemma throughout, both explicitly and implicitly.

Our main result in the one-to-one context is concerned with building as many pairwise internally-disjoint paths as we can from any block to any other block in the bipartite graph built using the $2$-step method (or, equivalently, from any node to any other node in the bipartite graph built using the $3$-step method). We explain the impact of the existence of these paths on the path diversity of subsequently built DCNs presently. One added and significant complication in the proof of the following result comes about when the transversal design $T$ is a $[k+1,k]$-transversal design (so, there is the potential for $\Delta=k+1 > k$ paths).

\begin{theorem}\label{thm:main}
Let $k,\Delta,d \geq 2$. Let $H$ be built by the $2$-step method from the $(d,\Delta)$-bipartite graph $H_0$ using the $[\Delta,k]$-transversal design $T$.

\begin{itemize}

\item[\emph{(\/}a\emph{)}] Let $Q$ and $Q^\prime$ be distinct blocks of $H_0$ so that there are $\lambda\geq 1$ pairwise internally-disjoint paths in $H_0$ from $Q$ to $Q^\prime$, each of length at most $\mu$. There are $\min\{\Delta,k\}$ pairwise internally-disjoint paths from any block $B_{Q,V}$ of $H$ to any other block $B_{Q^\prime,V^\prime}$ of $H$. Furthermore, if $\lambda \geq 2$ then there are $\Delta$ pairwise internally-disjoint paths from any block $B_{Q,V}$ of $H$ to any other block $B_{Q^\prime,V^\prime}$ of $H$. All paths have length at most $\mu+4$.

\item[\emph{(\/}b\emph{)}] If $B_{Q,V}$ and $B_{Q,V^\prime}$ are distinct blocks of $H$ then there are $\Delta$ pairwise internally-disjoint paths from $B_{Q,V}$ to $B_{Q,V^\prime}$, each of length at most $6$ and lying entirely within $T_Q$.

\end{itemize}
\end{theorem}

\begin{proof}Recall that we mentioned in Section~\ref{sec:tdesigns} that necessarily $\Delta\leq k+1$.\smallskip

\noindent\underline{Case (\emph{a\/})(\emph{i\/})}: Suppose that: $\Delta = k+1$; $\lambda\geq 2$; and the distinct nodes $p_1$ and $p_2$ are common neighbours in $H_0$ of $Q$ and $Q^\prime$.\smallskip

\noindent We `batch' the groups of nodes of $T_Q$ and $T_{Q^\prime}$ together so that in each of $T_Q$ and $T_{Q^\prime}$, the $k+1$ groups of nodes form $1$ batch of $k$ groups and $1$ batch of $1$ group as follows:
\begin{itemize}
\item for $i\in\{1,2\}$, define $G_0^i=G_{p_i}=H_0^i$
\item the remaining $k-1$ groups within $T_Q$ are $G_1^1,G_2^1,\ldots,G_{k-1}^1$ and the remaining $k-1$ groups within $T_{Q^\prime}$ are $H_1^1,H_2^1,\ldots,H_{k-1}^1$ so that:
\begin{itemize}
\item any group of the form $G_j^1$, where $j>0$, is associated with some node $p\not\in\{p_1,p_2\}$ of $H_0$ that is adjacent to both $Q$ and $Q^\prime$ if, and only if, the group $H_j^1$ is associated with the same node $p$ of $H_0$ (so, if $G_j^1$ and $H_j^1$ are associated with the same node $p\not\in\{p_1,p_2\}$ of $H_0$ then they are the same group in $H$).
\end{itemize}
\end{itemize}

For each $j\in\{0,1,\ldots,k-1\}$, let $r_j^1\in G_j^1$ (resp. $s_j^1\in H_j^1$) be the unique node of $G_j^1$ (resp. $H_j^1$) that is adjacent to $B_{Q,V}$ (resp. $B_{Q^\prime,V^\prime}$) in $H$. Note that the pair $r_j^1$ and $s_j^1$ lie in the same group of nodes in $H$ if, and only if, both $G_j^1$ and $H_j^1$ are associated with the same node $p$ of $H_0$ and this node $p$ is adjacent to both $Q$ and $Q^\prime$ in $H_0$. The situation can be visualized as in Fig.~\refstepcounter{fig}\thefig\label{basicsetupcaseai} (where in this case $Q$ and $Q^\prime$ have $a+2\geq 2$ common neighbours in $H_0$ and where, for example, $r_1^1\neq s_1^1$ but $r_a^1=s_a^1$).\smallskip

\begin{figure}[t]
\centering
\scalebox{0.68}[0.68]{
\includegraphics{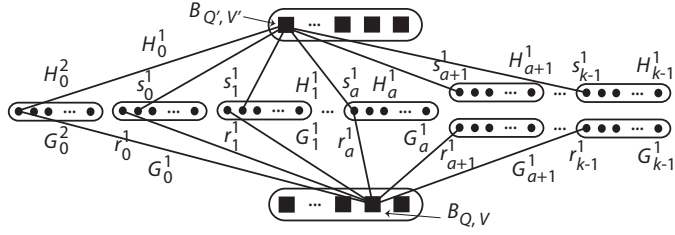}}
\caption{The basic set-up in Case (\emph{a\/})(\emph{i\/}).}
\end{figure}

Let $G_0^1=\{r_0^1,t_1,\ldots,t_{k-1}\}$ and $H_0^1=\{s_0^1,w_1,\ldots,w_{k-1}\}$ so that:
\begin{itemize}
\item if $r_0^1=s_0^1$ then $t_j=w_j$, for $j=1,2,\ldots,k-1$
\item if $r_0^1\neq s_0^1$ then $r_0^1=w_1$, $s_0^1=t_1$ and $t_j=w_j$, for $j=2,3,\ldots,k-1$.
\end{itemize}

We are now ready to generate some blocks within $T_Q$ and $T_{Q^\prime}$ in $H$. For each $j\in\{1,2,\ldots,k-1\}$:
\begin{itemize}
\item let $B_{r_j^1,t_j}$ be the unique block of $T_Q$ in $H$ generated by the nodes $r_j^1\in G_j^1$ and $t_j\in G_0^1$
\item let $B^\prime_{s_j^1,w_j}$ be the unique block of $T_{Q^\prime}$ in $H$ generated by the nodes $s_j^1\in H_j^1$ and $w_j\in H_0^1$.
\end{itemize}
So, we have generated $k-1$ blocks in $T_Q$ and $k-1$ blocks in $T_{Q^\prime}$. Note that any block of $T_Q$ is necessarily distinct from any block of $T_{Q^\prime}$. By Lemma~\ref{lem:genblocks} applied twice to both $T_Q$ and $T_{Q^\prime}$, all blocks of 
$\{B_{r_j^1,t_j} : j=1,2,\ldots,k-1\}$
are distinct and different from $B_{Q,V}$, and all blocks of $\{B^\prime_{s_j^1,w_j} : j=1,2,\ldots,k-1\}$
are distinct and different from $B_{Q^\prime,V^\prime}$. Call these two sets of blocks our working sets of blocks.

We are now in a position to build some paths from $B_{Q,V}$ to $B_{Q^\prime,V^\prime}$ in $H$. If $r_0^1=s_0^1$ then define the paths:
\begin{itemize}
\item $\pi_0^1$ as $B_{Q,V},r_0^1,B_{Q^\prime,V^\prime}$
\item $\pi_1^1$ as $B_{Q,V},r_1^1,B_{Q^\prime,V^\prime}$, if $r_1^1=s_1^1$, and as $B_{Q,V}, r_1^1,B_{r_1^1,t_1},t_1, B^\prime_{s_1^1,w_1},s_1^1,$\linebreak$B_{Q^\prime,V^\prime}$, if $r_1^1\neq s_1^1$ (note that $t_1=w_1$).
\end{itemize}
If $r_0^1\neq s_0^1$ then define the paths:
\begin{itemize}
\item $\pi_0^1$ as $B_{Q,V},r_0^1,B^\prime_{s_1^1,w_1},s_1^1, B_{Q^\prime,V^\prime}$ (note that $w_1=r_0^1$)
\item $\pi_1^1$ as $B_{Q,V},r_1^1,B_{r_1^1,t_1},s_0^1, B_{Q^\prime,V^\prime}$ (note that $t_1=s_0^1$).
\end{itemize}

We'll now build paths from $B_{Q,V}$ to $B_{Q^\prime,V^\prime}$ using nodes from the groups $\{G_0^1\}\cup\{G_j^1,H_j^1:j=2,3,\ldots,k-1\}$. For each $j\in\{2,3,\ldots,k-1\}$:
\begin{itemize}
\item if $r_j^1\neq s_j^1$ then define the path:
\begin{itemize}
\item $\pi^1_j$ as $B_{Q,V},r_j^1,B_{r_j^1,t_j},t_j,B^\prime_{s_j^1, w_j},s_j^1,B_{Q^\prime,V^\prime}$ (note that $t_j=w_j$)
\end{itemize}
\item if $r_j^1=s_j^1$ then define the path: \begin{itemize}
\item $\pi^1_j$ as $B_{Q,V},r_j^1,B_{Q^\prime,V^\prime}$.
\end{itemize}
\end{itemize}

Note that out of all of the $k$ `$\pi$-paths' constructed above, the only way that we can have that two of our paths are not internally-disjoint is when $r_0^1\neq s_0^1$ but $r_1^1=s_1^1$ (in which case $\pi_0^1$ and $\pi_1^1$ share the common node $r_1^1=s^1_1$). In this case, choose $x_1\in G_1^1\setminus\{r_1^1\}$. Let $B_{r_0^1,x_1}$ be the block of $T_Q$ in $H$ generated by $r_0^1(=w_1)\in G_0^1$ and $x_1\in G_1^1$, and let $B^\prime_{s_0^1,x_1}$ be the block of $T_{Q^\prime}$ in $H$ generated by $s_0^1(=t_1)\in G_0^1$ and $x_1\in G_1^1$ (in essence, we have dispensed with the blocks $B_{r_1^1,t_1}$ and $B^\prime_{s_1^1,w_1}$ and replaced them with the blocks $B_{r_0^1,x_1}$ and $B^\prime_{s_0^1,x_1}$ in our working sets of blocks). The conditions of Lemma~\ref{lem:genblocks} still hold and so the blocks in our working sets of blocks are all distinct and different from $B_{Q,V}$ and $B_{Q^\prime,V^\prime}$. Redefine the paths:
\begin{itemize}
\item $\pi_0^1$ as $B_{Q,V},r_1^1,B_{Q^\prime,V^\prime}$
\item $\pi_1^1$ as $B_{Q,V},r_0^1,B_{r_0^1,x_1},x_1, B^\prime_{s_0^1,x_1},s_0^1,B_{Q^\prime,V^\prime}$.
\end{itemize}
The paths from the resulting set of $k$ $\pi$-paths are now pairwise internally-disjoint and each has length at most $6$.

Let $r_0^2$ (resp. $s_0^2$) be the unique node of $G_0^2$ (resp. $H_0^2$) that is adjacent to $B_{Q,V}$ (resp. $B_{Q^\prime,V^\prime}$) in $H$. Suppose that $r_0^2= s_0^2$. In this case, we build the path $\pi_1^2$ defined as $B_{Q,V},r_0^2,B_{Q^\prime,V^\prime}$. This path is clearly internally-disjoint from all of the $k$ $\pi$-paths constructed above. Alternatively, suppose that $r_0^2\neq s_0^2$. If $k\geq 3$ then there is a node $x_1^1\in G_1^1\setminus\{r_1^1,s_1^1\}$. Let $B_{r_0^2,x_1^1}$ be the block of $T_Q$ within $H$ generated by $r_0^2$ and $x_1^1$, and let $B_{s_0^2,x_1^1}$ be the block of $T_{Q^\prime}$ within $H$ generated by $s_0^2$ and $x_1^1$. By Lemma~\ref{lem:genblocks}, these blocks are different from $B_{Q,V}$, $B_{Q^\prime,V^\prime}$ and all other blocks used within the $k$ $\pi$-paths constructed above (even when we make the amendments to our working sets of blocks as detailed in the preceding paragraph). Define the path $\pi_0^2$ as $B_{Q,V},r_0^2,B_{r_0^2,x_1^1},x_1^1, B_{s_0^2,x_1^1},s_0^2,B_{Q^\prime,V^\prime}$. This path has length $6$ and is clearly internally-disjoint from all of the $k$ $\pi$-paths constructed above. 

On the other hand, suppose that $k=2$; so, $\Delta=3$. In particular, a $[3,2]$-transversal design exists. We deal with this case from scratch.

\begin{lemma}
There is exactly one $[3,2]$-transversal design up to isomorphism and this is the transversal design depicted in Fig.~\refstepcounter{fig}\thefig\label{32trans}.
\end{lemma}

\begin{figure}[t]
\centering
\scalebox{0.68}[0.68]{
\includegraphics{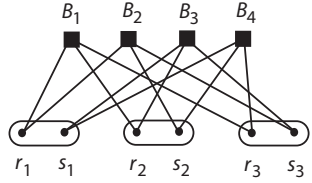}}
\caption{The unique $[3,2]$-transversal design.}
\end{figure}

\begin{proof}
In some $[3,2]$-transversal design, let the set of blocks be $\{B_1,B_2,B_3,B_4\}$ and let the group of nodes $G_i$ be $\{r_i,s_i\}$, for $i=1,2,3$. W.l.o.g., we must have the set of edges $$\{(r_1,B_1),(r_1,B_2),(s_1,B_3),(s_1,B_4),(r_2,B_1),(r_2,B_3),(s_2,B_2),(s_2,B_4)\}.$$ W.l.o.g., the node $B_4$ is adjacent to $r_3$, and $r_3$ is adjacent to one other block. The only possible block that $r_3$ can be adjacent to is $B_1$ (as otherwise we would have two nodes in different groups adjacent to $2$ distinct blocks).\end{proof}

Name the blocks and nodes of $T_Q$ as in Fig.~\ref{32trans}. W.l.o.g. suppose that $B_{Q,V}=B_1$ (it is easy to see that there is an automorphism of $T_Q$ mapping any block to any other block). 
There are two cases to consider: when $Q$ and $Q^\prime$ have $3$ common neighbours in $H_0$; and when they have only $2$ common neighbours in $H_0$. However, before we deal with these cases, choose any $3$ distinct nodes in $T_Q$. A tedious case-by-case analysis yields that no matter which $3$ nodes are chosen, there are $3$ pairwise internally-disjoint paths from $B_1$ to the $3$ nodes within $T_Q$. For example, suppose that the $3$ chosen nodes are $r_1$, $s_2$ and $s_3$. The $3$ paths are: $B_1,r_1$; $B_1,r_3,B_4,s_2$; and $B_1,r_2,B_3,s_3$. It turns out that if the $3$ chosen nodes are in $3$ different groups then the length of any path is at most $5$, whereas if the $3$ chosen nodes are in $2$ different groups then the length of any path is at most $3$.

Suppose that $Q$ and $Q^\prime$ have $3$ common neighbours in $H_0$. Choose the $3$ nodes in $T_Q$ as the neighbours of $B_{Q^\prime,V^\prime}$ in $T_{Q^\prime}$. Consequently, from above, we clearly obtain $3$ pairwise internally-disjoint paths from $B_{Q,V}$ to $B_{Q^\prime,V^\prime}$ as required. Moreover, each path has length at most $6$.

Suppose that $Q$ and $Q^\prime$ have only $2$ common neighbours in $H_0$ where the groups corresponding to these neighbours are $\{r_i,s_i\}$ and $\{r_j,s_j\}$, with $i,j\in\{1,2,3\}$, $i\neq j$. Choose $3$ nodes in $T_{Q^\prime}$ as the two neighbours of $B_{Q^\prime,V^\prime}$ in $\{r_i,s_i\}$ and $\{r_j,s_j\}$, call them $x_i$ and $x_j$, plus one other node, call it $x$, say, from one of these groups, with the remaining unchosen node from these two groups being denoted by $y$. By above, there are $3$ pairwise internally-disjoint paths, $\pi^\prime_1$, $\pi^\prime_2$ and $\pi^\prime_3$, in $T_{Q^\prime}$ from $B_{Q^\prime,V^\prime}$ to $x_i$, $x_j$ and $x$, and we may assume that these paths do not involve $y$; for if one does then it must be the path to $x$, in which case we simply choose $y$ as our third chosen node, above, instead of $x$ (it is not difficult to see that the path to $x$ has length at most $3$). By above, there are also $3$ pairwise internally-disjoint paths, $\pi_1$, $\pi_2$ and $\pi_3$, in $T_Q$ from $B_1$ to $x_i$, $x_j$ and $x$, each of length at most $3$; moreover, these paths do not share any nodes with the paths $\pi^\prime_1$, $\pi^\prime_2$ and $\pi^\prime_3$ apart from the end-nodes. Consequently, we clearly obtain $3$ pairwise internally-disjoint paths from $B_{Q,V}$ to $B_{Q^\prime,V^\prime}$ as required. Moreover, each path has length at most $6$.\smallskip

\noindent\underline{Case (\emph{a\/})(\emph{ii\/})}: Suppose that: $\Delta = k+1$; $\lambda\geq 2$; and there is exactly one common neighbour in $H_0$ of $Q$ and $Q^\prime$, namely the node $p_1$.\smallskip

\noindent As $\lambda\geq 2$, there is a path in $H_0$ of the form $Q,q_1,Q_1,q_2,Q_2,\ldots,Q_{m-1},q_{m},Q^\prime$ where $m\leq \frac{\mu}{2}$ and where $p_1$ does not appear on this path (note that $q_1\neq q_m$). We `batch' our groups similarly to as we did before:
\begin{itemize}
\item define $G_0^1=G_{p_1}=H_0^1$, $G_0^2=G_{q_1}$ and $H_0^2=G_{q_m}$
\item the remaining $k-1$ groups within $T_Q$ are $G_1^1,G_2^1,\ldots,G_{k-1}^1$ and the remaining $k-1$ groups in $T_{Q^\prime}$ are $H_1^1,H_2^1,\ldots,H_{k-1}^1$.
\end{itemize}
Note that we necessarily have that the groups $G_j^1$ and $H_j^1$ are distinct, for $j\in \{1,2,\ldots,k-1\}$ (as are the groups $G_0^2$ and $H_0^2$).

For each $j\in\{0,1,\ldots,k-1\}$, let $r_j^1\in G_j^1$ (resp. $s_j^1\in H_j^1$) be the unique node of $G_j^1$ (resp. $H_j^1$) that is adjacent to $B_{Q,V}$ (resp. $B_{Q^\prime,V^\prime}$) in $H$. Also, let $r_0^2 \in G_0^2$ (resp. $s_0^2 \in H_0^2$) be the unique node of $G_0^2$ (resp. $H_0^2$) that is adjacent to $B_{Q,V}$ (resp. $B_{Q^\prime,V^\prime}$) in $H$.

We construct the paths $\pi_j^1$, for $j=0,1,\ldots,k-1$, exactly as we did in Case (\emph{a\/})(\emph{i\/}). In addition, define the paths $\eta_0^2$ as $B_{Q,V},r_0^2$ and $\nu_0^2$ as $B_{Q^\prime,V^\prime},s_0^2$. If we can find a path in $H$ from $r_0^2$ to $s_0^2$ so that no node or block of this path, apart from the nodes and blocks of $\eta_0^2$ and $\nu_0^2$, lies in $T_Q$ or $T_{Q^\prime}$ then we are done. In $H$: there is a path of length $2$ lying entirely within $T_{Q_1}$ so that the source is $r_0^2$ and the destination is some node $y_2 \in G_{q_2}$; there is a path of length $2$ lying entirely within $T_{Q_2}$ so that the source is $y_2 \in G_{q_2}$ and the destination is some node $y_3 \in G_{q_3}$; $\ldots$; there is a path of length $2$ lying entirely within $T_{Q_{m-1}}$ so that the source is $y_{m-1}\in G_{q_{m-1}}$ and the destination is $s_0^2 \in H_0^2$. We clearly have a required path of length at most $\mu$. So, we have constructed $\Delta=k+1$ pairwise internally-disjoint paths from $B_{Q,V}$ to $B_{Q^\prime,V^\prime}$ so that $k$ of these paths have length at most $6$ and the remaining path has length at most $\mu$.\smallskip

\noindent\underline{Case (\emph{a\/})(\emph{iii\/})}: Suppose that: $\Delta = k+1$; $\lambda\geq 2$; and there are no common neighbours in $H_0$ of $Q$ and $Q^\prime$.\smallskip

\noindent As $\lambda\geq 2$, there are paths in $H_0$ of the form $Q,q_1,Q_1,q_2,Q_2,\ldots,Q_{a-1},q_a,Q^\prime$ and $Q,p_1,P_1,p_2,P_2,\ldots,P_{b-1},p_b,Q^\prime$ where $a,b\leq \frac{\mu}{2}$ and where these paths are internally-disjoint.

We `batch' our groups similarly to as we did before:
\begin{itemize}
\item define $G_0^1=G_{q_1}$ and $H_0^1=G_{q_a}$
\item choose $k-1$ groups within $T_Q$ (different from $G_0^1$) as $G_1^1,G_2^1,\ldots,G_{k-1}^1$ and choose $k-1$ groups in $T_{Q^\prime}$ (different from $H_0^1$) as $H_1^1,H_2^1,\ldots,H_{k-1}^1$.
\end{itemize}
Note that we necessarily have that the groups $G_j^1$ and $H_j^1$ are distinct, for $j\in \{0,1,\ldots,k-1\}$.

For each $j\in\{0,1,\ldots,k-1\}$, let $r_j^1\in G_j^1$ (resp. $s_j^1\in H_j^1$) be the unique node of $G_j^1$ (resp. $H_j^1$) that is adjacent to $B_{Q,V}$ (resp. $B_{Q^\prime,V^\prime}$) in $H$. Let $G_0^1=\{r_0^1,t_1,\ldots,t_{k-1}\}$ and $H_0^1=\{s_0^1,w_1,\ldots,w_{k-1}\}$. For each $j\in\{1,2,\ldots,k-1\}$, let $B_{r_j^1,t_j}$ be the block of $T_Q$ generated by $r_j^1\in G_j^1$ and $t_j\in G_0^1$, and let $B^\prime_{s_j^1,w_j}$ be the block of $T_{Q^\prime}$ generated by $s_j^1\in H_j^1$ and $w_j\in H_0^1$. By Lemma~\ref{lem:genblocks} applied twice to both $T_Q$ and $T_{Q^\prime}$, all blocks of 
$\{B_{r_j^1,t_j} : j=1,2,\ldots,k-1\}$
are distinct and different from $B_{Q,V}$, and all blocks of $\{B^\prime_{s_j^1,w_j} : j=1,2,\ldots,k-1\}$
are distinct and different from $B_{Q^\prime,V^\prime}$.

For $j\in\{1,2,\ldots,k-1\}$, let $\eta_j^1$ be the path $B_{Q,V},r_j^1,B_{r_j^1,t_j},t_j$ and let $\nu^1_j$ be the path $B_{Q^\prime,V^\prime},s_j^1,B^\prime_{s_j^1,w_j}, w_j$. Define the path $\eta_0^1$ as $B_{Q,V},r_0^1$ and the path $\nu_0^1$ as $B_{Q^\prime,V^\prime},s_0^1$. In $H$: there are $k$ paths of length $2$ from the nodes $r_0^1,t_1,\ldots,t_{k-1}$ of $G_{q_1}$ to distinct nodes $y_0^2, y_1^2,\ldots,y_{k-1}^2$ of $G_{q_2}$, respectively, so that all blocks on these paths lie in $T_{Q_1}$ and are distinct; there are $k$ paths of length $2$ from the nodes $y_0^2, y_1^2,\ldots,y_{k-1}^2$ of $G_{q_2}$ to distinct nodes $y_0^3, y_1^3,\ldots,y_{k-1}^3$ of $G_{q_3}$, respectively, so that all blocks on these paths lie in $T_{Q_2}$ and are distinct; $\ldots$; and there are $k$ paths of length $2$ from $y_0^{a-1}, y_1^{a-1},\ldots,y_{k-1}^{a-1}$ to the nodes $s_0^1,w_2,\ldots,w_{k-1}$ of $G_{q_a}$, respectively, so that all blocks on these paths lie in $T_{Q_{m-1}}$ and are distinct. We can clearly piece all of the paths together to obtain $k$ pairwise internally-disjoint paths from $B_{Q,V}$ to $B_{Q^\prime,V^\prime}$ so that each path has length at most $\mu+4$.

We can build another path from $B_{Q,V}$ to $B_{Q^\prime,V^\prime}$ that is internally-disjoint from the $k$ paths just constructed by proceeding exactly as we did above or in Case (\emph{a\/})(\emph{ii\/}), corresponding to the alternative path from $Q$ to $Q^\prime$ in $H_0$. This path has length at most $\mu$.\smallskip

\noindent\noindent\underline{Case (\emph{a\/})(\emph{iv\/})}: Suppose that $\lambda=1$ or $\Delta\leq k$.\smallskip

\noindent By choosing the appropriate construction from the cases above, depending upon whether there is a common neighbour of $Q$ and $Q^\prime$ in $H_0$, we can clearly construct $\min\{\Delta,k\}$ pairwise internally-disjoint paths from $B_{Q,V}$ to $B_{Q^\prime,V^\prime}$ so that: if there is a common neighbour of $Q$ and $Q^\prime$ in $H_0$, all paths have length at most $6$; and if there is no common neighbour of $Q$ and $Q^\prime$ in $H$, all paths have length at most $\mu+4$.\smallskip

\noindent\underline{Case (\emph{b\/})}: Consider the case when our two blocks are $B_{Q,V}$ and $B_{Q,V^\prime}$. Suppose that the block $Q$ of $H_0$ is adjacent to the nodes $p_1,p_2,\ldots,p_\Delta$. For each $i\in\{1,2,\ldots,\Delta\}$, let $r_i\in G_{p_i}$ be adjacent to $B_{Q,V}$ in $H$ and let $s_i\in G_{p_i}$ be adjacent to $B_{Q,V^\prime}$ in $H$. W.l.o.g. suppose that $r_i\neq s_i$, for $i=1,2,\ldots,b$, and that $r_i=s_i$, for $i=b+1,b+2,\ldots,\Delta$. 

Suppose that $b\geq 2$. For each $i\in\{1,2,\ldots,b-1\}$, let $B_{r_i,s_{i+1}}$ be the block of $T_Q$ that is generated by $r_i$ and $s_{i+1}$, and let $B_{r_b,s_1}$ be the block of $T_Q$ that is generated by $r_b$ and $s_{1}$. By Lemma~\ref{lem:genblocks}, all blocks $B_{r_1,s_2},B_{r_2,s_3},\ldots,B_{r_{b-1},s_b},B_{r_b,s_1}$ are distinct and different from $B_{Q,V}$ and $B_{Q,V^\prime}$. Hence: if $\pi_i$ is the path $B_{Q,V},r_i,B_{r_i,s_{i+1}},s_{i+1},B_{Q,V^\prime}$, for $i\in\{1,2,\ldots,b-1\}$; if $\pi_b$ is the path $B_{Q,V},r_b,$\linebreak$B_{r_b,s_1},s_1,B_{Q,V^\prime}$; and if $\pi_i$ is the path $B_{Q,V},r_i,B_{Q,V^\prime}$, for $i\in\{b+1,b+2,\ldots,\Delta\}$, then paths in the resulting set are pairwise internally-disjoint, with each path having length at most $4$.

If $b=0$ then the above construction trivially yields $\Delta$ paths of length $2$ from $B_{Q,V}$ to $B_{Q,V^\prime}$. Suppose that $b=1$. Choose $x_2\in G_{p_2}\setminus\{r_2\}$ and let $B_{r_1,x_2}$ (resp. $B_{s_1,x_2}$) be the block of $T_Q$ generated by $r_1$ and $x_2$ (resp. $s_1$ and $x_2$). Clearly, $B_{r_1,x_2}$, $B_{s_1,x_2}$, $B_{Q,V}$ and $B_{Q,V^\prime}$ are all distinct. So, if $\pi_1$ is the path $B_{Q,V},r_1,B_{r_1,x_2},x_2,B_{s_1,x_2},s_1, B_{Q,V^\prime}$ and $\pi_i$ is the path $B_{Q,V},r_i,B_{Q,V^\prime}$, for $i\in\{2,3,\ldots,\Delta\}$, then we obtain $\Delta$ pairwise internally-disjoint paths, with all paths having length $2$ except one which has length $6$.\end{proof}

Theorem~\ref{thm:main} is clearly optimal in the sense that the maximal number of pairwise internally-disjoint paths is always constructed (this follows from a simple application of Menger's Theorem). Also, irrespective of the erroneous proofs in \cite{QFZ15}, Theorem~\ref{thm:main}(\emph{b\/}) extends the claimed results in \cite{QFZ15} by deriving $\Delta$ pairwise internally-disjoint paths from any block $B_{Q,V}$ in $H$ to any block $B_{Q,V^\prime}$ (this scenario was not dealt with in \cite{QFZ15}). Note also that the chance to obtain more than $\min\{\Delta,k\}$ pairwise internally-disjoint paths comes about when we force $\Delta=k+1$ and choose a $[k+1,k]$-transversal design (if one exists).

Of course, Theorem~\ref{thm:main} yields path diversity in any DCN constructed using the $3$-step method with Methods $A$ and $B$. Suppose that Method $A$ has been used to construct a DCN where the number of server-nodes adjacent to some level-$1$ switch-node is at most the number of level-$2$ switch-nodes adjacent to the level-$1$ switch-node. If all level-$1$ switch-nodes are non-blocking then we can simultaneously facilitate data transfers from all the server-nodes adjacent to some level-$1$ switch-node to all the server-nodes adjacent to any other level-$1$ switch-node (in fact, we need only that the source and destination level-$1$ switch-nodes are non-blocking; all other level-$1$ switch-nodes can be blocking).

\subsection{Applying our construction}\label{sec:applyingcons}

In this section, we apply Theorem~\ref{thm:main} and p[rovide ome concrete illustrations of how we can obtain switch-centric DCNs that have the same diameter as Fat-Tree yet have more server-nodes and significant one-to-one path diversity. 

The primary difficulty in the proof of Theorem~\ref{thm:main} is in dealing with when the $[\Delta,k]$-transversal design is such that $\Delta=k+1$ (recall, $k,\Delta,d\geq 2$). However, dealing with this difficulty is worth it as having the capability to use $[k+1,k]$-transversal designs when applying the construction means that we obtain more flexibility as to the number of switch ports necessarily required in the resulting DCNs, as we illustrate now. In what follows, we limit ourselves (on the grounds of practicality) to switch-nodes with at most $128$ ports. If we were only to use $[\Delta,k]$-transversal designs where $(\Delta,k)\in\{(3,3),(4,4),(5,5),(7,7),(8,8),(9,9),$\linebreak$(11,11)\}$ (note that each of these $[\Delta,k]$-transversal designs exists; see Section~\ref{sec:tdesigns}) in the (one-iteration) $3$-step method then (assuming that we use bipartite graphs $H_0$ that have the same number of nodes as blocks; that is, for which $d=\Delta$) we need level-$2$ switch-nodes with $9$, $16$, $25$, $36$, $49$, $64$, $81$, $100$ or $121$ ports. If we allow $[\Delta,k]$-transversal designs where $(\Delta,k)\in\{(3,2),(4,3),(5,4),(6,5),(8,7),(9,8),(10,9)\}$ (again, note that each of these $[\Delta,k]$-transversal designs exists; see Section~\ref{sec:tdesigns}) then we have added flexibility in that we can also build DCNs with level-$2$ switch-nodes with $6$, $12$, $20$, $30$, $56$, $72$ or $90$ ports; of course, to ensure that we obtain full path diversity, we need that $H_0$ has at least $2$ internally-disjoint paths joining any two distinct blocks.

As regards finding large, regular, uniform bipartite graphs of line-diameter $4$ and so that there are at least $2$ internally-disjoint paths joining any two distinct blocks, this is not as straightforward as it is if we drop the second stipulation. There is an extensive literature as regards the construction of regular, uniform bipartite graphs of a given degree and where the degree is equal to the rank (see, for example, \cite{MS13}) but in so far as we are aware, the construction of such graphs with any added stipulations (relating to connectivity, for example)  has not been considered. Nevertheless, there are simple constructions that enable us to apply Theorem~\ref{thm:main} to the full, as we now illustrate. 

From \cite{MS13}, there is a regular, uniform bipartite graph of degree and rank $7$ with $173$ nodes and $173$ blocks, and which has graph-theoretic diameter $4$. Enumerate the nodes as $n_1,n_2,\ldots,n_{173}$ and the blocks as $b_1,b_2,\ldots,b_{173}$. Take two disjoint copies of this graph and add $346$ edges joining $n_i$ in one graph to $b_i$ in the other graph; moreover, the nodes (resp. blocks) of the new bipartite graph are exactly the nodes (resp. blocks) of the original disjoint copies. The resulting graph is a regular, uniform bipartite graph of degree and rank $8$ with graph-theoretic diameter at most $5$; furthermore, there are clearly at least $2$ internally-disjoint paths joining any pair of distinct blocks or any pair of distinct nodes where these paths have length at most $6$. Take this bipartite graph as $H_0$.

Apply the $3$-step method using a $[8,7]$-transversal design. This results in a bipartite graph with $16,954$ nodes, $2,442$ blocks, degree $8$, and rank $56$. Now apply Method $A$ with $c = 4$ and we obtain a DCN of diameter $6$ and with $406,896$ server-nodes, $16,954$ level-$1$ switch-nodes, $9,688$ level-$2$ switch-nodes and so that all switch-nodes have $56$ ports. By Theorem~\ref{thm:main}, there are paths from the $24$ server-nodes adjacent to the same level-$1$ switch-node $X$ to the $24$ server-nodes adjacent to another level-$1$ switch-node $Y$ so that the only switch-nodes that lie on more than one of these paths are $X$ and $Y$ and so that the length of each of these paths is at most $10$. In addition, we have spare capacity at the level-$1$ switch-nodes $X$ and $Y$ as $8$ links to level-$2$ switch-nodes are not used.

Alternatively (for an increase in the number of server-nodes incorporated and in path diversity but so that more ports are required on switch-nodes), apply the $3$-step method using a $[8,8]$-transversal design. This results in a bipartite graph with $22,144$ nodes, $2,768$ blocks, degree $8$, and rank $64$. Now apply Method $A$ with $c = 4$ and we obtain a DCN of diameter $6$ and with $708,608$ server-nodes, $22,144$ level-$1$ switch-nodes, $11,072$ level-$2$ switch-nodes and so that all switch-nodes have $64$ ports. By Theorem~\ref{thm:main}, there are paths from the $32$ server-nodes adjacent to the same level-$1$ switch-node $X$ to the $32$ server-nodes adjacent to another level-$1$ switch-node $Y$ so that the only switch-nodes that lie on more than one of these paths are $X$ and $Y$ and so that the length of each of these paths is at most $10$. The actual construction used will be dominated by the available hardware; that is, numbers of server-nodes and switch-nodes and the radix of switch-nodes.

Undertaking more iterations of the $2$-step construction before building our DCNs yields that if we use $[\Delta,k]$-transversal designs where $(\Delta,k)\in\{(3,3),(4,$\linebreak$4),(5,5),(7,7),(8,8),(9,9),(11,11)\}$ then we need level-$2$ switch-nodes with\linebreak$3.3.3=27$, $4.4.4=64$, $3.3.3.3=81$ or $5.5.5=125$ ports; and if we use $[\Delta,k]$-transversal designs where $(\Delta,k)\in\{(3,2),(4,3),(5,4),(6,5),(8,7),(9,8),$\linebreak$(10,9)\}$ then we need level-$2$ switch-nodes with an alternative range of port numbers. As an illustration, iterating the $2$-step method by mixing the use of $[3,3]$- and $[3,2]$-transversal designs, we can build DCNs where the level-$2$ switch-nodes need $6$, $9$, $12$, $18$, $24$, $27$, $36$, $48$, $54$, $72$, $81$, $96$ and $108$ ports. What is more, by Theorem~\ref{thm:main}, any bipartite graph built using the $2$-step method iterated more than once necessarily has maximal path diversity (as applying the $2$-step method once always yields a bipartite graph where there are at least $2$ internally-disjoint paths joining any two distinct blocks).

It has already been established in \cite{QFZ15} that the $2$-step and $3$-step methodologies are viable when it comes to building switch-centric DCNs that can host more server-nodes than a Fat-Tree and retain an acceptable level of (one-to-one) path diversity whilst maintaining a diameter of $6$; we further cement this viability in this paper. An important point to note is that we need not choose our bipartite graph $H_0$ to be as large as we can; as we have shown, smaller bipartite graphs might yield DCNs with a sufficiently large number of server-nodes and optimal one-to-one path diversity.

\section{One-to-many path diversity}\label{sec:onetomanypathdiv}

We now work towards building $\Delta$ pairwise edge-disjoint paths from any block in some bipartite graph $H$ built using the $2$-step method to the blocks of any given multi-set of $\Delta$ blocks (so, there might possibly be repeated blocks; here, $H$ and $\Delta$ are as in the statement of Theorem~\ref{thm:main} but where $\Delta\leq k$). Henceforth, when we write `set' we often mean `multi-set'. We begin by working only within some transversal design.

\begin{theorem}\label{thm:startpaths}
Let $T$ be any $[\Delta,k]$-transversal design where $k,\Delta\geq 2$ and where $\Delta\leq k$. Let $U$ be any block and let $t_1,t_2,\ldots,t_\Delta$ be any $\Delta$ nodes or blocks, called target-nodes or target-blocks, as appropriate, where there may be repetitions. For each $i=1,2,\ldots,\Delta$, there is a path $\pi_i$ from $U$ to $t_i$ of length at most $7$ so that these paths are pairwise edge-disjoint.
\end{theorem}

\begin{proof}
For each group of nodes $D_j$ within $T$, where $j\in\{1,2,\ldots,\Delta\}$, let $r_j$ be the (unique) node of $D_j$ adjacent to the block $U$; we call the nodes $r_1,r_2,\ldots,r_\Delta$ \emph{root-nodes\/}. Consider some group $D_j$. There may be target-nodes that are identical to the root-node $r_j$; call these target-nodes \emph{rooted\/}, with the remaining target-nodes in $D_j$ called \emph{non-rooted\/}. Call the number of rooted target-nodes in $D_j$ the \emph{multiplicity\/} of the root-node $r_j$.

There are two essential cases: (\emph{a\/}) we have $\Delta$ target-nodes and no target-blocks; (\emph{b\/}) we have at least $1$ target-block.\smallskip

\noindent\underline{Case (\emph{a\/})}: Suppose that we have $\Delta$ target-nodes and no target-blocks.\smallskip

\noindent We rank the groups of $T$ as $D_{n_1},D_{n_2},\ldots,D_{n_\Delta}$ in decreasing order of the number of occurrences of non-rooted target nodes within the group, with ties broken according to decreasing multiplicity of the root-nodes (and then arbitrarily). We attempt to match the non-rooted target-nodes in $D_{n_1}$ with the root-nodes $r_{n_2},r_{n_3},\ldots,r_{n_\Delta}$ in this order but only if the root-node has multiplicity $0$ (that is, we skip over root-nodes of non-zero multiplicity; note that any skipped root-node is identical to at least $1$ target-node). If we are successful then we attempt to match the non-rooted target-nodes in $D_{n_2}$ by continuing down our list of root-nodes (again, skipping over root-nodes of non-zero multiplicity). If we are successful then we attempt to match the non-rooted target-nodes in $D_{n_3}$, and so on. There are three possibilities.
\begin{itemize}
\item[(\emph{i\/})] We successfully match every non-rooted target-node without running out of root-nodes (of multiplicity $0$). This happens when $r_{n_1}$ has non-zero multiplicity or when there is a root-node with multiplicity at least $2$.

\item[(\emph{ii\/})] We successfully match all but one of the non-rooted target-nodes and the final non-rooted target-node does not lie in $D_{n_1}$, in which case we match this target-node with $r_{n_1}$. This happens when $r_{n_1}$ has multiplicity $0$, every root-node has multiplicity at most $1$, and there is a non-rooted target-node that does not lie in $D_{n_1}$.

\item[(\emph{iii\/})] We have one non-rooted target-node of $D_{n_1}$ remaining to be matched and also the root-node $r_{n_1}$ unmatched. This happens when all of the non-rooted target-nodes lie in $D_{n_1}$ and $r_{n_1}$ has multiplicity $0$.
\end{itemize}

Consider Sub-case (\emph{ii\/}). We extend our matching so that every root-node of multiplicity $1$ is matched with the unique target-node that is identical to it. We have a complete matching of root-nodes to target-nodes so that no target-node is matched with the root-node in its own group unless the target-node is (the unique target-node) identical to the root-node. For every pair $(r,t)$ where $r$ is a root-node matched with a target-node $t$ and so that $r$ and $t$ do not lie in the same group, let $U_{r,t}$ be the block generated by $r$ and $t$. Call the resulting set of blocks the $U$-blocks. By Lemma~\ref{lem:genblocks}, all of the $U$-blocks are distinct and different from $U$. If $U_{r,t}$ is a $U$-block then define the path $\pi_r$ as $U,r,U_{r,t},t$; and if the target-node $t$ is identical to the root-node $r$ then define the path $\pi_r$ as $U,t$. The resulting $\Delta$ paths are pairwise internally-disjoint.

Consider Sub-case (\emph{iii\/}).  We have an almost complete matching of root-nodes to target-nodes so that no target-node is matched with the root-node in its own group, except that some target-node $t^\prime$ of $D_{n_1}$ is not matched and nor is the root-node $r_{n_1}$. As we did above, we generate a set of $U$-blocks, one for each matched-pair. Again, these $U$-blocks are all distinct and different from $U$, and by proceeding as above we obtain $\Delta-1$ pairwise internally-disjoint paths from $U$ to target-nodes.

Consider $t^\prime$ and $r_{n_1}$. As $\Delta\geq 2$, there is some node $x$ in the group $D_{n_2}$ that is neither a root-node nor a target-node. Let $U^\prime_{r_{n_1},x}$ (resp. $U^\prime_{x,t^\prime}$) be the block generated by $r_{n_1}$ and $x$ (resp. $x$ and $t^\prime$). By Lemma~\ref{lem:genblocks}, $U^\prime_{r_{n_1},x}$ is different from $U$ and every $U$-block; also, $U^\prime_{x,t^\prime}$ is different from $U$ and $U^\prime_{r_{n_1},x}$. However, it could be that $U^\prime_{x,t^\prime}$ is identical to some $U$-block (for this to happen we would need that $t^\prime$ is identical to some other target-node). If $t$ is the target-node of $D_{n_1}$ matched with $r_{n_2}$ then $U^\prime_{x,t^\prime}$ is different from $U_{r_{n_2},t}$. Hence, there are at most $\Delta-2$ $U$-blocks with which $U^\prime_{x,t^\prime}$ might be identical. As we have at least $\Delta-1$ choices for $x$ in $D_{n_2}$ (recall, $\Delta\leq k$), we can always choose $x$ so that $U^\prime_{x,t^\prime}$ is different from every $U$-block. Define the path $\pi_{r_{n_1}}$ as $U,r_{n_1},U^\prime_{r_{n_1},x},x, U^\prime_{x,t^\prime},t^\prime$. The resulting $\Delta$ paths from $U$ to the target-nodes are pairwise internally-disjoint.

Consider Sub-case (\emph{i\/}). We can extend our matching so that every root-node of non-zero multiplicity is matched with one target-node that is identical to it. Hence, we have a partial matching of root-nodes to target-nodes so that no target-node is matched with the root-node in its own group unless the target-node is identical to the root-node. As we did above, we generate a set of $U$-blocks, one for each matched-pair where the root-node in the pair is different from its matched target-node. Again, these $U$-blocks are all distinct and different from $U$, and we obtain pairwise internally-disjoint paths from $U$ to all of the target-nodes involved. We also obtain paths of length $1$ from $U$ to every target-node that is identical to a root-node and has been matched with it. If there are no root-nodes of multiplicity greater than $1$ then the resulting $\Delta$ paths are pairwise internally-disjoint and we are done. So, suppose that we have paths $\pi_1,\pi_2,\ldots,\pi_{\Delta-b}$ that are pairwise internally-disjoint and that there are $a\geq 1$ root-nodes of multiplicity at least $2$ with $b$ unmatched root-nodes (so, $b$ is the number of target-nodes remaining to be dealt with; of course, $b\geq a$). Note that any group in which some hitherto unmatched root-node lies, apart from $D_{n_1}$ if $r_{n_1}$ is still unmatched (that is, has multiplicity $0$), contains no target-nodes (because of the order in which we initially match target-nodes to root-nodes) and the groups containing unmatched root-nodes are either $D_{n_{\Delta-b+1}},D_{n_{\Delta-b+2}},\ldots,D_{n_{\Delta}}$, if $r_{n_1}$ is matched, or $D_{n_1},D_{n_{\Delta-b+2}},D_{n_{\Delta-b+3}},\ldots,D_{n_{\Delta}}$, if $r_{n_1}$ is unmatched.

Suppose that $b=1$; hence, there is exactly one root-node $r_c$, where $c\leq \Delta-1$, of multiplicity greater than $1$ and this multiplicity is $2$. W.l.o.g. let the solitary target-node remaining to be dealt with be $t_2$ (which is identical to both $r_c$ and some other target-node $t_1$), with the solitary root-node remaining to be dealt with being either $r_{n_\Delta}$ or $r_{n_{1}}$, as appropriate. If $\Delta=2$ then we must have $\{r_{n_1},x_{n_1}\}\subseteq D_{n_1}$ and $\{r_{n_2},x_{n_2}\}\subseteq D_{n_2}$ with $x_{n_1}\neq r_{n_1}$ and $x_{n_2}\neq r_{n_2}$ so that the two target nodes $t_1$ and $t_2$ are both equal to $r_{n_1}$ (note that in this case we define no $U$-blocks). Let $U^\prime_{x_{n_2},r_{n_1}}$ (resp. $U^\prime_{r_{n_2},x_{n_1}}$, $U^\prime_{x_{n_1},x_{n_2}}$) be the block generated by $x_{n_2}$ and $r_{n_1}$ (resp. $r_{n_2}$ and $x_{n_1}$, $x_{n_1}$ and $x_{n_2}$). The blocks $U$, $U^\prime_{x_{n_2},r_{n_1}}$, $U^\prime_{r_{n_2},x_{n_1}}$ and $U^\prime_{x_{n_1},x_{n_2}}$ are all distinct. Define the path $\pi_2$ as $U, r_{n_2},U^\prime_{r_{n_2},x_{n_1}},x_{n_1}, U^\prime_{x_{n_1},x_{n_2}}, x_{n_2}, U^\prime_{x_{n_2},r_{n_1}}, t_2$ and the path $\pi_1$ as $U,t_1$; the two paths are internally-disjoint and we are done.

Alternatively, suppose that $b=1$ and $\Delta\geq 3$ (and so $k\geq 3$). If $c=n_{\Delta-1}$ then there is a non-rooted target-node in each $D_j$, for $j\in\{n_1,n_2,\ldots,n_{\Delta-2}\}$, with the unmatched root-node being $r_{n_1}$. Choose $x\in D_\Delta\setminus\{r_{n_\Delta}\}$. Otherwise, if $c\neq n_{\Delta-1}$ then $D_{n_{\Delta-1}}$ contains at most $1$ target-node, which, if it exists, is rooted, with the unmatched root-node being either $r_{n_1}$ or $r_{n_\Delta}$. Choose $x\in D_{\Delta-1}\setminus\{r_{n_{\Delta-1}}\}$. Whichever is the case, let $r$ be the unmatched root-node (and so $r\in\{r_{n_1},r_{n_\Delta}\}$). Let $U^\prime_{x,t_2}$ (resp. $U^\prime_{r,x}$) be the block generated by $x$ and $t_2$ (resp. $r$ and $x$). By Lemma~\ref{lem:genblocks}, the $U$-blocks, $U$, $U^\prime_{x,t_2}$ and $U^\prime_{r,x}$ are all distinct. Define the path $\pi_\Delta$ as $U,r,U^\prime_{r,x},x,U^\prime_{x,t_2},t_2$ so as to obtain $\Delta$ pairwise internally-disjoint paths from $U$ to the target-nodes; hence, we are done.

Now suppose that $b\geq 2$ (note that $b\leq \Delta-1\leq k-1$). As stated above, the root-nodes remaining to be dealt with are either $r_{n_{\Delta - b +1}},r_{n_{\Delta - b +2}},\ldots,r_{n_\Delta}$ or $r_{n_1},r_{n_{\Delta - b +2}},r_{n_{\Delta - b +3}},\ldots,r_{n_\Delta}$. Suppose that the root-nodes remaining to be dealt with are $r_{n_{\Delta - b +1}},r_{n_{\Delta - b +2}},\ldots,r_{n_\Delta}$ and the target-nodes remaining to be dealt with are $t_1,t_2,\ldots,t_b$ (of course, every such target-node is identical to an already matched root-node). For each $i\in\{\Delta-b+1,\Delta-b+2,\ldots,\Delta\}$, let $D_{n_i}=\{r_{n_i},x_2^{n_i},x_3^{n_i},\ldots,x_k^{n_i}\}$ and choose $x_{n_i}^\prime\in D_{n_i}\setminus\{r_{n_i}\}$ (from our earlier remark, there are no target-nodes in $D_{n_i}$). For each $i\in\{\Delta-b+1,\Delta-b+2,\ldots,\Delta-1\}$, let $U^\prime_{r_{n_i},x_{n_{i+1}}^\prime}$ be the block generated by $r_{n_i}$ and $x_{n_{i+1}}^\prime$, and let $U^\prime_{r_{n_\Delta},x_{n_{\Delta-b+1}}^\prime}$ be the block generated by $r_{n_\Delta}$ and $x_{n_{\Delta-b+1}}^\prime$; call these blocks the $U^\prime$-blocks. By Lemma~\ref{lem:genblocks}, the $U^\prime$-blocks are distinct and each $U^\prime$-block is different from every $U$-block and $U$. For each $i\in\{\Delta-b+1,\Delta-b+2,\ldots,\Delta\}$, let $\bar{U}_{x_{n_i}^\prime,t_i}$ be the block generated by $x_{n_i}^\prime$ and $t_i$; call these blocks the $\bar{U}$-blocks. By Lemma~\ref{lem:genblocks}, each $\bar{U}$-block is different from $U$, every $U$-block and from every $U^\prime$-block (note that any $t_i$ is a root-node and so not adjacent to any $U$-block or $U^\prime$-block). However, it is possible that $\bar{U}_{x_{n_i}^\prime,t_i} = \bar{U}_{x_{n_j}^\prime,t_j}$, for $i\neq j$ (for this to happen we would need that $t_i=t_j$, as otherwise we would have two root-nodes adjacent to both $U$ and another block). Note that for each $i\in\{\Delta-b+1,\Delta-b+2,\ldots,\Delta\}$: we have $k-1$ possible choices within $D_{n_i}$ for $x_{n_i}^\prime$; and for $j_1,j_2\in\{2,3,\ldots,k\}$, where $j_1\neq j_2$, the block $\bar{U}_{x_{j_1}^{n_i},t_i}$, generated by $x_{j_1}^{n_i}$ and $t_i$, is different from the block $\bar{U}_{x_{j_2}^{n_i},t_i}$, generated by $x_{j_2}^{n_i}$ and $t_i$. 

Choose $x_{n_{\Delta-b+1}}^\prime=x_2^{n_{\Delta-b+1}}$ and $x_{n_{\Delta-b+2}}^\prime=x_2^{n_{\Delta-b+2}}$. Suppose we have that $\bar{U}_{x_{n_{\Delta-b+2}}^\prime,t_2} = \bar{U}_{x_{n_{\Delta-b+1}}^\prime,t_1}$; if so then re-choose $x_{n_{\Delta-b+2}}^\prime=x_3^{n_{\Delta-b+2}}$. Necessarily, $\bar{U}_{x_{n_{\Delta-b+2}}^\prime,t_2} \neq \bar{U}_{x_{n_{\Delta-b+1}}^\prime,t_1}$. Choose $x_{n_{\Delta-b+3}}^\prime=x_2^{n_{\Delta-b+3}}$. Suppose that $\bar{U}_{x_{n_{\Delta-b+3}}^\prime,t_3} = \bar{U}_{x_{n_{\Delta-b+1}}^\prime,t_1}$; if so then re-choose $x_{n_{\Delta-b+3}}^\prime=x_3^{n_{\Delta-b+3}}$. Suppose that $\bar{U}_{x_{n_{\Delta-b+3}}^\prime,t_3} = \bar{U}_{x_{n_{\Delta-b+2}}^\prime,t_2}$; if so then re-choose $x_{n_{\Delta-b+3}}^\prime=x_4^{n_{\Delta-b+3}}$. Necessarily, $\bar{U}_{x_{n_{\Delta-b+1}}^\prime,t_1},\bar{U}_{x_{n_{\Delta-b+2}}^\prime,t_2},\bar{U}_{x_{n_{\Delta-b+3}}^\prime,t_3}$ are distinct. Proceed in this way until $x_{n_{\Delta-b+2}}^\prime,x_{n_{\Delta-b+3}}^\prime, \ldots, x_{n_\Delta}^\prime$ have been chosen. Note that as $b\leq \Delta-1\leq k-1$, the above procedure can always be completed. What results is the set of distinct blocks $\{\bar{U}_{x_{n_{\Delta-b+i}}^\prime,t_i}: i=1,2,\ldots,b\}$. For each $i=1,2,\ldots,b-1$, define the path $\pi_{\Delta-b+i}$ as $U,r_{n_{\Delta-b+i}},U^\prime_{r_{n_{\Delta-b+i}},x_{n_{\Delta-b+i +1}}^\prime}, x_{n_{\Delta-b+i+1}}^\prime, \bar{U}_{x_{\Delta-b+i+1}^\prime, t_i},t_i$, and define the path $\pi_\Delta$ as $U,r_{n_\Delta},U^\prime_{r_{n_\Delta},x_{n_{\Delta-b+1}}^\prime}, x_{n_{\Delta-b+1}}^\prime, \bar{U}_{x_{n_{\Delta-b+1}}^\prime, t_b},t_b$. The resulting $\Delta$ paths $\pi_1,\pi_2,\ldots,\pi_\Delta$ from $U$ to the target-nodes are pairwise internally-disjoint.

Alternatively, suppose that the root-nodes remaining to be dealt with are $r_{n_1},r_{n_{\Delta - b +2}},r_{n_{\Delta - b +3}},\ldots,r_{n_\Delta}$. We proceed exactly as above except that we start from a node $x_{n_1}^\prime\in D_{n_1}\setminus\{r_{n_1}\}$ that is different from any target-node (such a node exists). We obtain our pairwise internally-disjoint paths as before.
\smallskip

\noindent\underline{Case (\emph{b\/})}: Suppose that there is at least $1$ target-block.\smallskip

\noindent W.l.o.g. we may assume that the $a$ target-nodes $t_1,t_2,\ldots,t_a$, where $0\leq a \leq \Delta-1$, lie within the groups $D_1,D_2,\ldots,D_a$ and that the target-blocks are $U_1,U_2\ldots,U_{\Delta-a}$. Suppose that some target-block $U_i$ is adjacent to some root-node $r_j$ of some group $D_j$, where $i\in\{1,2,\ldots,\Delta-a\}$ and $j\in\{a+1,a+2,\ldots,\Delta\}$. Remove the target-block $U_{i}$ (temporarily) from our set of targets and include the new target-node $r_{j}$. Iterate this process. Hence, w.l.o.g. we may assume that: our target-nodes are the original target-nodes $t_1,t_2,\ldots,t_a$ along with the new target-nodes $r_{a+1},r_{a+2},\ldots,r_{a+b}$, where each new target-node $r_{a+i}$ is adjacent to the now removed old target-block $U_i$; and our target-blocks are $U_{b+1},U_{b+2},\ldots,U_{\Delta-a}$ with none of these target-blocks adjacent to any root-node in the groups $D_{a+b+1},D_{a+b+2}\ldots,D_\Delta$. For each $i\in\{1,2,\ldots,\Delta-a\}$: let the node $x_{a+b+i}\in D_{a+b+i}\setminus\{r_{a+b+i}\}$ be adjacent to $U_{b+i}$; and (temporarily) remove the target-block $U_{b+i}$ and include the new target-node $x_{a+b+i}$.

Apply the construction in Case (\emph{a\/}) to our new set of $\Delta$ target-nodes. We obtain $\Delta$ paths, one from $U$ to each of our target-nodes so that these paths are internally-disjoint. Consider some new target-node $r_{a+i}$, where $i\in\{1,2,\ldots,b\}$. By the construction of our paths, the path corresponding to this new target-node is $U,r_{a+i}$ and $r_{a+i}$ does not appear on any other path (there is no repetition of $r_{a+i}$ in our set of target-nodes). Extend the path $U,r_{a+i}$ to the path $U,r_{a+i},U_i$, for $i=1,2,\ldots,b$. Consider some new target-node $x_{a+b+i}$, where $i\in\{1,2,\ldots,\Delta-a\}$. Suppose that the edge $(U_{b+i},x_{a+b+i})$ appears on some path. By the construction of our paths, the only way that this can happen is if this edge is the last edge on the path from $U$ to $x_{a+b+i}$. If this is the case then truncate this path at $U_{b+i}$. Alternatively, if the edge $(U_{b+i},x_{a+b+i})$ does not appear on some path then we extend the path from $U$ to $x_{a+b+i}$ by the addition of the edge to $U_{b+i}$. Consequently, we obtain a set of paths from $U$ to each of our original target-nodes and target-blocks so that these paths are pairwise edge-disjoint. Note that: target-nodes only appear as destinations and apart from possibly target-nodes, no node appears on more than one path; and no block appears on more than one path except possibly for some target-blocks (which might appear as internal nodes on paths). The result follows.\end{proof}

Note that the construction in Theorem~\ref{thm:startpaths} is weaker than those in the previous section as we obtain only that paths are pairwise edge-disjoint rather than pairwise internally-disjoint. However, we do obtain the following result as an immediate corollary of the construction in Theorem~\ref{thm:startpaths}. 

\begin{corollary}
Let $T$ be any $[\Delta,k]$-transversal design where $k,\Delta\geq 2$ and where $\Delta\leq k$. Let $U$ be any block and let $t_1,t_2,\ldots,t_\Delta$ be any $\Delta$ nodes, called target-nodes, where there may be repetitions. For each $i=1,2,\ldots,\Delta$, there is a path $\pi_i$ from $U$ to $t_i$ of length at most $7$, so that the paths $\pi_1,\pi_2,\ldots,\pi_\Delta$ are pairwise internally-disjoint.
\end{corollary}

We now build some many-to-many edge-disjoint paths within some transversal design.

\begin{theorem}\label{thm:endpaths}
Let $T$ be any $[\Delta,k]$-transversal design where $\Delta \leq k$ and $k,\Delta\geq 2$. Let $a + b = \Delta_0\leq \Delta$ where $a,b\geq 0$. Suppose that we are given $a$ nodes, the target-nodes, and $b$ blocks, the target-blocks, so that there might be repetitions amongst the target-nodes and target-blocks. Suppose that $D_0$ is a group of nodes that contains no target-nodes. There exists a set $S$ of $\Delta_0$ distinct nodes of $D_0$ such that there are $\Delta_0$ pairwise internally-disjoint paths, each of length at most $3$, the sources of which are the nodes of $S$ and the destinations of which are all the target-nodes and target-blocks.
\end{theorem}

\begin{proof}
Suppose that $b\geq 1$ (we'll deal with the case when $b=0$ later)  and suppose that the distinct target-blocks are $U_1,U_2,\ldots,U_c$, so that the target-blocks $U_{c+1},U_{c+2},\ldots,U_b$ all lie in $\{U_i:i=1,2,\ldots,c\}$. Furthermore, suppose that for each $i\in\{1,2,\ldots,c\}$,  the target-block $U_i$ is repeated $n_i$ times in the set of target-blocks. So, $b=\sum\limits_{i=1}^c n_i$. We define that $U_i\equiv U_j$, for $i,j\in\{1,2,\ldots,c\}$ if, and only if, $U_i$ and $U_j$ are adjacent to the same node of $D_0$. Let $U_1,U_2,\ldots,U_d$ (where $d\geq 1$) be representatives from the resulting equivalence classes (so, $d\leq c$) and let $x_1^i$ be the node of $D_0$ adjacent to $U_i$, for $i=1,2,\ldots,d$. Thus, we immediately obtain $d$ paths $\pi_1^1,\pi_1^2,\ldots,\pi_1^d$ of length $1$ from distinct nodes of $D_0$ to the target-blocks $U_1,U_2,\ldots,U_d$.

For ease of notation, we rename some of the groups of nodes of $T$ as $\{D_0\}\cup\{D_j^i:i=1,2,\ldots,d; j=2,3,\ldots,n_i\}\cup\{D_j^i:i=d+1,d+2,\ldots,c; j=1,2,\ldots,n_i\}$ so that no target-node lies in any of these groups (note that the number of such groups is $(\sum\limits_{i=1}^cn_i)-d+1 \leq b=\Delta_0-a$ and so this is possible). For each $i\in\{1,2,\ldots,d\}$ and each $j\in\{2,3,\ldots,n_i\}$, choose $x^i_j\in D_0\setminus\{x_1^1,x_1^2,\ldots,x_1^d\}$, and for each $i\in\{d+1,d+2,\ldots,c\}$ and each $j\in\{1,2,\ldots,n_i\}$, choose $x^i_j\in D_0\setminus\{x_1^1,x_1^2,\ldots,x_1^d\}$, so that all chosen nodes are distinct. For each $i\in\{1,2,\ldots,d\}$ and each $j\in\{2,3,\ldots,n_i\}$, let $r_j^i\in D_j^i$ be the unique node adjacent to $U_i$, and for each $i\in\{d+1,d+2,\ldots,c\}$ and each $j\in\{1,2,\ldots,n_i\}$, let $r_j^i\in D_j^i$ be the unique node adjacent to $U_i$.

For each $i\in\{1,2,\ldots,d\}$ and each $j\in\{2,3,\ldots,n_i\}$, let $U_j^i$ be the block generated by $x_j^i$ and $r_j^i$, and for each $i\in\{d+1,d+2,\ldots,c\}$ and each $j\in\{1,2,\ldots,n_i\}$, let $U_j^i$ be the block generated by $x_j^i$ and $r_j^i$. Call the resulting blocks generated the $U$-blocks. In particular, as every $U$-block is adjacent to a different node of $D_0$, all $U$-blocks are distinct. Moreover, as no target-block is adjacent to the same node of $D_0$ that any $U$-block is adjacent to, every $U$-block is different from every target-block. For each $i\in\{1,2,\ldots,d\}$ and each $j\in\{2,3,\ldots,n_i\}$, define the path $\pi_j^i$ as $x_j^i,U_j^i,r_j^i,U_i$, and for each $i\in\{d+1,d+2,\ldots,c\}$ and each $j\in\{1,2,\ldots,n_i\}$, define the path $\pi_j^i$ as $x_j^i,U_j^i,r_j^i,U_i$. The paths from the set $\{\pi_j^i:i=1,2,\ldots,c; j=1,2,\ldots,n_i\}$ are clearly internally-disjoint.

Write $n_0=a$ and suppose that the target-nodes are $t_1,t_2,\ldots,t_{n_0}$. Let $x_1^0,x_2^0,\ldots,x_{n_0}^0$ be distinct nodes of $D_0\setminus\{x_j^i:i=1,2,\ldots,c; j=1,2,\ldots,n_i\}$. For each $j\in\{1,2,\ldots,n_0\}$, let $U^0_j$ be the block generated by $x_j^0$ and $t_j$. As above, all such blocks are distinct and different from any block generated so far. For each $j\in\{1,2,\ldots,n_0\}$, define the path $\pi_j^0$ as $x_j^0,U_j^0,t_j$. The resulting set of paths $\{\pi_j^i:i=0,1,\ldots,c; j=1,2,\ldots,n_i\}$ is as required.

Alternatively, if $b=0$ then we dispense with the above construction of paths to target-blocks and proceed identically as regards the target-nodes. The result clearly follows. \end{proof}

We are now in a position to use Theorems~\ref{thm:startpaths}~and~\ref{thm:endpaths} to obtain the main result of this section.

\begin{theorem}\label{thm:mainonetomany}
Let $k, \Delta, d\geq 2$ so that $\Delta\leq k$. Let $H$ be built by the $2$-step method applied to the connected $(d,\Delta)$-bipartite graph $H_0$ using the $[\Delta,k]$-transversal design $T$. Let $B$ be some block of $H$ and let $B_1,B_2,\ldots,B_\Delta$ be blocks of $H$ that are not necessarily distinct but different from $B$. There exists paths from $B$ to $B_1,B_2,\ldots,B_\Delta$ so that no edge of $H$ appears in more than one of these paths. 
\end{theorem}

\begin{proof}
Let $Q_1,Q_2,\ldots,Q_q$ be the exact distinct blocks of $H_0$ such that $\cup_{i=1}^qT_{Q_i}$ contains the blocks $B_1,B_2,\ldots,B_\Delta$ within $H$ (in particular, $q\leq \Delta$), and let $Q_0$ be the block of $H_0$ such that $T_{Q_0}$ contains the block $B$ within $H$. Let $Z$ be a tree within $H_0$ that is rooted at $Q_0$ and is such that: every block of $\{Q_i:i=1,2,\ldots,q\}$ appears in $Z$; and all leaves of $Z$ are blocks within $\{Q_i:i=1,2,\ldots,q\}$. We use the tree $Z$ as a skeleton so as to build our required paths in $H$.

Call the blocks $B_1,B_2,\ldots,B_\Delta$ the \emph{$H$-target-blocks\/}. Label every node $p$ (resp. block $Q$) in $Z$ with a non-negative integer $\mu(p)$ (resp. $\mu(Q)$) detailing the number of $H$-target-blocks that are associated with a block of $Z$ that is a descendant of $p$ (resp. a descendent of $Q$ or with $Q$ itself). So, for example, the root $Q_0$ is such that $\mu(Q_0)=\Delta$ and any leaf (block) $Q$ of $Z$ is such that $\mu(Q)$ is the number of $H$-target-blocks within $T_Q$.

Suppose that $p$ is some node of $Z$ whose children are all leaves (and so blocks). Suppose that w.l.o.g. these children are $Q_1,Q_2,\ldots,Q_r$. For each $i\in\{1,2,\ldots,r\}$, by Theorem~\ref{thm:endpaths}, there exists a set $S_i$ of $\mu(Q_i)$ nodes of the group of nodes of $H$ associated with the node $p$ of $H_0$ so that there are $\mu(Q_i)$ pairwise internally-disjoint paths from the nodes of $S_i$ to the $H$-target-blocks associated with $Q_i$ where each of these paths has length at most $3$ (note that the edges of these paths lie in $T_{Q_i}$; of course, the edges of $T_{Q_i}$ are disjoint from the edges of $T_{Q_j}$, for any $i\neq j$). Consequently, we obtain a multi-set $S_p=\cup_{i=1}^r S_i$ of $\mu(p)$ nodes in the group of nodes in $H$ associated with the node $p$ of $H_0$ so that there are $\mu(p)$ paths in $\cup_{i=1}^rT_{Q_i}$ from the nodes of $S_p$ to the $H$-target-blocks associated with the blocks $Q_1,Q_2,\ldots,Q_r$. These $\mu(p)$ paths are pairwise internally-disjoint but they might have common sources.

Suppose that $Q$ is some non-root block of $Z$ whose children are w.l.o.g.  $p_1,p_2,\ldots,p_r$ and so that the following holds:
\begin{itemize}
\item associated with each child $p_i$ is a multi-set $S_i$ of $\mu(p_i)$ nodes in the group of nodes of $H$ associated with the node $p_i$ of $H_0$
\item for each child $p_i$, there are $\mu(p_i)$ paths from the nodes of $S_i$ to the $H$-target-blocks associated with blocks that are descendants of $p_i$ in $T$ so that all of these paths have length at most $l$
\item no edge of $H$ appears in more than one of the $\sum\limits_{i=1}^r\mu(p_i)$ paths that are associated with some child of $Q$.
\end{itemize}
Let $p_0$ be the node of $Z$ that is the parent of $Q$. By Theorem~\ref{thm:endpaths}, there is a set $S_0$ of $\mu(p_0)$ nodes in the group of nodes of $H$ associated with $p_0$ together with $\mu(p_0)$ paths from the nodes of $S_0$ to the nodes of $\cup_{i=1}^rS_i$ in union with the $H$-target-blocks associated with $Q$ where the paths are pairwise internally-disjoint and each path has length at most $3$. Hence, by concatenating the paths involved, we have $\mu(p_0)=\mu(Q)$ paths from the nodes of $S_0$ to the $H$-target-blocks associated with all descendant blocks of $p_0$ in $Z$ where no edge of $H$ appears in more than one of these paths and the length of any of these paths is at most $l+3$.

Finally, suppose that the children of $Q_0$ in $Z$ are w.l.o.g. $p_1,p_2,\ldots,p_r$ and are such that the following holds:
\begin{itemize}
\item associated with each child $p_i$ is a multi-set $S_i$ of $\mu(p_i)$ nodes in the group of nodes of $H$ associated with the node $p_i$ of $H_0$
\item for each child $p_i$, there are $\mu(p_i)$ paths from the nodes of $S_i$ to the $H$-target-blocks associated with blocks that are descendants of $p_i$ in $T$ so that all of these paths have length at most $l$
\item no edge of $H$ appears in more than one of the $\sum\limits_{i=1}^r\mu(p_i)$ paths that are associated with some child of $Q_0$.
\end{itemize}
By Theorem~\ref{thm:startpaths}, we obtain paths from $B$ to the nodes of $\cup_{i=1}^r S_i$ in union with the multi-set of blocks associated with $Q_0$ so that no edge of $H$ appears in more than one of these paths and all paths have length at most $7$. Consequently, by concatenating paths, we obtain $\Delta$ paths from $B$ to $B_1,B_2,\ldots,B_\Delta$ so that no edge of $H$ appears in more than one of these paths and the paths have length at most $l+7$. The result follows by induction. Moreover, it is easy to see that if the depth of $Z$ is $h$ then the length of the longest of these paths is at most $3\frac{h}{2}+7$.
\end{proof}

We have two remarks as regards Theorem~\ref{thm:mainonetomany}: first, note the additional bound of $3\frac{h}{2}+7$ on the lengths of the paths derived in the proof of Theorem~\ref{thm:mainonetomany} in terms of the height $h$ of the tree $Z$; and, second, this theorem is weaker than Theorem~\ref{thm:main} in that in Theorem~\ref{thm:mainonetomany} the paths constructed are pairwise edge-disjoint rather than pairwise internally-disjoint as they are in Theorem~\ref{thm:main}.

Of course, armed with the constructions of switch-centric DCNs from Section~\ref{sec:composition}, it should be clear how we can obtain pairwise edge-disjoint paths joining all the server-nodes adjacent to some level-$1$ switch-node in some appropriately constructed DCN to any identically-sized set of distinct server-nodes (irrespective of whether they are adjacent to different level-$1$ switch-nodes), so long as the number of server-nodes adjacent to some level-$1$ switch-node is no more than the number of level-$2$ switch-nodes adjacent to it.

\section{Conclusion}

In this paper, we have shown how combinatorial design theory can be used to build switch-centric DCNs of diameter at most $6$ and with many more server-nodes than the Fat-Tree DCN but so that there is still considerable one-to-one and one-to-many path diversity. We regard the more general demonstration that combinatorial mathematics can enhance the design of modern-day computational infrastructures such as data centres as one of the primary contributions of this paper. Whilst we have demonstrated that combinatorial mathematics has the potential to add to and improve the design of DCNs, the DCNs obtained by our constructions need to be studied in much greater detail with regard to the numerous other properties that a switch-centric DCN has to have in order to make it viable as a practical DCN. For example: although we bound the diameter of our DCNs, we need to derive (optimal) routing algorithms (within bipartite graphs built using the $2$-step method) so as to meet these bounds; and (as was noted in \cite{QFZ15}, it would be beneficial if the bisection width of the DCNs constructed in this paper could be ascertained (bisection width is often used as a proxy for throughput in DCNs).

Our results also throw up some directions for further research and we mention three such directions now.

It would be interesting to discover more mechanisms for converting bipartite graphs constructed using the $2$-step method into DCNs than those developed in \cite{QFZ15} and detailed in Section~\ref{sec:composition}. We envisage that such a study would go hand-in-hand with research into building DCNs which possess yet more beneficial properties as regards their efficacy as DCNs (as highlighted above).

As we mention in Section~\ref{sec:applyingcons}, our constructions have drawn attention to a hitherto unstudied problem in combinatorics namely the construction of regular, uniform bipartite graphs with additional properties such as having at least $2$ internally-disjoint paths joining any two blocks. It would be interesting to study problems such as this in a solely mathematical context.

Our results have hinted that the study of transversal designs as bipartite graphs and in a graph-theoretic context is worth pursuing. For example, if one looks at Theorem~\ref{thm:main} then there are $\Delta$ pairwise internally-disjoint paths, each of length at most $6$, joining any two distinct blocks in some transversal design $T$; and if one looks at Theorem~\ref{thm:startpaths} then, if $\Delta\leq k$, there are $\Delta$ pairwise edge-disjoint paths, each of length at most $6$, joining any given source block with any given multi-set of $\Delta$ target blocks in some transversal design $T$. Such results might be of independent interest within some appropriate network context. Within a DCN $N$ built using the $2$-step method, there are many `copies' of the bipartite graph corresponding to the chosen transversal design. These copies and the above constructions might be utilized where there is traffic localization, e.g., in a virtualization context where many guest DCNs are embedded within the DCN $N$.



\begin{thebibliography}{99}

\bibitem{AMW10} D. Abts, M.R. Marty, P.M. Wells, P. Klausler and H. Liu, ``Energy Proportional Datacenter Networks'', \emph{Proc. of \emph{37\/}th Ann. Int. Symp. on Computer Architecture\/}, 2010, pp. 338--347.

\bibitem{ABD09} J.H. Ahn, N. Binkert, A. Davis, M. McLaren and R.S. Schreiber, ``HyperX: Topology, Routing, and Packaging of Efficient Large-scale Networks'', \emph{Proc. of Conf. on High Performance Computing Networking, Storage and Analysis\/}, 2009, article no. 41.

\bibitem{ALV08} M. Al-Fares, A. Loukissas and A. Vahdat, ``A Scalable, Commodity Data Center Network Architecture'', \emph{SIGCOMM Computer Communication Review\/}, vol. 38, no. 4, 2008, pp. 63--74.

\bibitem{BBD95} J.C. Bermond, J. Bond and S. Djelloul, ``Dense Bus Networks of Diameter 2'', \emph{Proc. Workshop on Interconnection Networks\/}, DIMACS Ser., vol. 21, Annals New York Academy of Sciences, 1995, pp. 9--18. 

\bibitem{CDS99} C.J. Colbourn, J.H. Dinitz and D.R. Stinson, ``Applications of Combinatorial Designs to Communications, Cryptography and Networking'', Surveys in Combinatorics (ed. J.D. Lamb and D.A. Preece), Cambridge University Press, 1999, pp. 37--100.

\bibitem{DG04} J. Dean and S. Ghemawat, ``MapReduce: Simplified Data Processing on Large Clusters'', \emph{Proc. of \emph{6}th Symp. on Operating System Design and Implementation\/}, 2004, pp. 137--150.

\bibitem{Die10} R. Diestel, \emph{Graph Theory\/}, Springer, 2010.

\bibitem{GHJ09}  A. Greenberg, J.R. Hamilton, N. Jain, S. Kandula, C. Kim, P. Lahiri, D.A. Maltz, P. Patel and S. Sengupta, ``VL2: A Scalable and Flexible Data Center Network'', \emph{SIGCOMM Computer Communication Review\/}, vol. 39, no. 4, 2009, pp. 51--62.

\bibitem{HSM06} B. Heller, S. Seetharaman, P. Mahadevan, Y. Yiakoumis, P. Sharma, S. Banerjee and N. McKeown, ``ElasticTree: Saving Energy in Data Center Networks'', \emph{Proc. \emph{7\/}th USENIX Conf. on Networked Systems Design and Implementation\/}'', 2006, pp. 249--264.

\bibitem{LMV13} Y. Liu, J.K. Muppala, M. Veeraraghavan, D. Lin and J. Katz, \emph{Data Centre Networks: Topologies, Architectures and Fault-Tolerance Characteristics\/}, Springer, 2013.

\bibitem{MS13} M. Miller and J. Siran, ``Moore graphs and beyond: a survey of the degree/diameter problem'', \emph{Electronic Journal of Combinatorics\/}, vol. 20, no. 2, 2013, article DS14.

\bibitem{MPF09} R.N. Mysore, A. Pamboris, N. Farrington, N. Huang, P. Miri, S. Radhakrishnan, V. Subramanya and A. Vahdat, ``Portland: A Scalable Fault-tolerant Layer 2 Data Center Network Fabric'', \emph{SIGCOMM Computer Communication Review\/}, vol. 39, no. 4, 2009, pp. 39--50.

\bibitem{QFZ15} G. Qu, Z. Fang, J. Zhang  and S.-Q. Zheng, ``Switch-centric Data Center Network Structures based on Hypergraphs and Combinatorial Block Designs'', \emph{IEEE Transactions on Parallel and Distributed Systems\/}, vol. 26, no. 4, 2015, pp. 1154--1164.

\bibitem{Sti04} D.R. Stinson, \emph{Combinatorial Designs: Constructions and Analysis\/}, Spring-\linebreak er, 2004.

\bibitem{WXN12} K. Wu, J. Xiao and L.M. Ni, ``Rethinking the Architecture Design of Data Center Networks'', \emph{Frontiers of Computer Science\/}, vol. 6, no. 5, 2012, pp. 596--603.


\end{thebibliography}
\end{document}